\documentclass[conference]{IEEEtran}

\usepackage[utf8]{inputenc}
\usepackage[T1]{fontenc}
\usepackage[english]{babel}
\usepackage{amsmath,amssymb,amsfonts}
\usepackage{amsthm}
\usepackage{bm}
\usepackage{mathtools}
\usepackage{graphicx}
\usepackage{xcolor}
\usepackage{tikz}
\usepackage{cite}
\usepackage{pgfplots}
\usetikzlibrary{arrows.meta,calc,positioning}
\usepackage{xcolor}
\definecolor{revisionblue}{RGB}{0,70,140}

\newtheorem{theorem}{Theorem}
\newtheorem{proposition}{Proposition}
\newtheorem{lemma}{Lemma}
\newtheorem{corollary}{Corollary}
\newtheorem{definition}{Definition}
\newtheorem{assumption}{Assumption}
\newtheorem{remark}{Remark}
\newtheorem{problem}{Problem}

\newcommand{\R}{\mathbb{R}}
\newcommand{\V}{\mathcal{V}}
\newcommand{\E}{\mathcal{E}}
\newcommand{\G}{\mathcal{G}}
\newcommand{\N}{\mathcal{N}}
\newcommand{\U}{\mathcal{U}}
\newcommand{\Y}{\mathcal{Y}}
\newcommand{\dist}{\operatorname{dist}}
\newcommand{\proj}{\operatorname{proj}}
\newcommand{\col}{\operatorname{col}}
\newcommand{\argmin}{\operatorname*{arg\,min}}
\newcommand{\argmax}{\operatorname*{arg\,max}}
\newcommand{\supp}{\sigma}

\begin{document}

\title{Distributed Containment of a Compromised Agent through Repulsive Cages}

\author{
Luigi Petruzziello$^1$, Camilla Fioravanti$^{1*}$, and Gabriele Oliva$^{1}$%
\thanks{$^1$Department of Engineering, University Campus Bio-Medico of Rome,
Via Alvaro del Portillo, 21 - 00128 Roma, Italy.}%
\thanks{$^*$Corresponding author. Email: c.fioravanti@unicampus.it}%
}

\maketitle

\begin{abstract}
UAV swarms and cyber-physical multi-agent systems are increasingly deployed in safety-critical missions that require coordinated motion, distributed decision making, and autonomy. A major security risk arises when a legitimate agent is hijacked and driven by adversarial high-level commands. Rather than focusing on detection and isolation of malicious agents, we exploit a structural property common in autonomous platforms: low-level collision-avoidance modules are typically implemented as independent safety layers and may remain active even under high-level compromise.
Building on this property, we propose a distributed containment framework that uses the compromised agent’s uncompromised avoidance response as an indirect actuation channel. Defender agents select their geometric configuration to shape the repulsive field experienced by the target, with the goal of keeping it inside a prescribed admissible region and, when required, steering it toward a desired destination. The interaction is modeled as an online Stackelberg game in which defenders act as leaders and the adversary reacts by choosing the target command.
Using support-function and normal-cone arguments, we derive an exact geometric characterization of robust one-step containment and introduce the notion of a repulsive cage. These results define a centralized Stackelberg oracle and motivate a fully distributed online approximation based on local communication and dynamic field estimation. We prove sublinear dynamic-regret bounds with respect to the centralized benchmark, quantifying the effect of network-induced estimation errors and temporal variability of the stage-wise optimum. Simulations validate the approach and corroborate the theory.

\end{abstract}


{\color{black}
\section{Introduction}
\label{sec:introduction}
Cooperative multi-agent systems are now a consolidated platform in several application domains, ranging from environmental surveillance and autonomous vehicle formations to infrastructure monitoring and search-and-rescue operations~\cite{OlfatiSaberFaxMurray2007,CortesMartinezKaratasBullo2004}. Many of these platforms operate according to a two-layer control architecture: a high-level command channel, which generates mission setpoints from the coordination logic, and a low-level safety layer, typically implemented in hardware or firmware, which automatically handles critical functions such as collision avoidance through pairwise repulsive actions based on relative distances~\cite{Khatib1986,RimonKoditschek1992}. The latter is designed to be robust and independent of the mission channel, precisely because it acts as a non-negotiable safety primitive.

This architectural separation has an important security implication: a compromise of the high-level command channel, although severe, does not automatically imply control over the entire control stack of the affected agent. This is precisely the scenario considered in this work. A legitimate agent of the fleet is hijacked by an external attacker, who gains control of its mission channel and freely selects its motion setpoint, but cannot disable or rewrite the low-level layer. A motivating example is a formation of unmanned aerial vehicles in which one vehicle is hijacked: the vehicle, which we refer to as the \emph{target}, continues to automatically generate its repulsive response with respect to the other agents, which we refer to as \emph{defenders}. The compromise is therefore hybrid: the high-level command is under adversarial control, while the low-level response remains unaltered, and this asymmetry is the structure exploited by our approach.

The objective is to keep the target inside the admissible region $\Omega_{k+1}$ at each time step, despite the adversarial choice of $u_a$. The defenders cannot command the target and do not know $u_a$ a priori; their only leverage is their collective configuration. By choosing where to move, they shape the aggregate repulsive field $g_k(y)$ that the target's low-level layer applies to the target itself. The resulting containment problem is therefore \emph{indirect} and has a \emph{field-shaping} nature: the defenders do not control the target directly, but shape the repulsive environment in which it evolves. Compared with the classical use of artificial potential fields~\cite{Khatib1986,RimonKoditschek1992}, where repulsion is a local collision-avoidance tool for cooperative agents, here the same primitive is repurposed as a control channel acting on an agent whose mission channel is adversarial.

The informational asymmetry between the defenders and the adversary naturally lends itself to a hierarchical formulation in the spirit of Stackelberg games~\cite{vonStackelberg2011}: the defenders first commit to a configuration $y \in \mathcal{Y}_k$, observable through the physical geometry of the network, and the adversary then reacts by selecting $u_a$. At each time step, the configuration-command pair defines a stage game, and the overall trajectory is a sequence of such games coupled through the dynamics. Unlike the literature on \emph{security games}~\cite{Tambe2011,Paruchuri2008}, where the main object is the computation of a stationary Stackelberg equilibrium, here we are interested in online tracking of a sequence of stage equilibria in a dynamical system with time-varying constraints and distributed information. From this perspective, our work is also related to no-regret learning for Stackelberg games~\cite{BalcanBlumHaghtalabProcaccia2015}, while differing from it in two key aspects: the leader's action space is continuous and geometric rather than discrete, and the performance is measured through dynamic regret against a sequence of stage oracles rather than against a fixed strategy in hindsight. Methodologically, our analysis builds on the framework of online convex optimization~\cite{Zinkevich2003,ShalevShwartz2012} and networked distributed optimization~\cite{NedicOzdaglar2009}, with the closest reference point being the recent literature on distributed online convex optimization with time-varying constraints~\cite{YiLiXieJohansson2020}.

The closest line of work to our problem is \emph{indirect herding}~\cite{ChipadePanagou2019,ChipadePanagou2020,ChipadeMarellaPanagou2021,PiersonSchwager2018,LicitraBellDixon2019,AulettaFioreRichardsonDiBernardo2022,SebastianMontijanoSagues2022,ZhangLeiDuanPengPan2024,LamaDiBernardo2024,LamaDiBernardoKlapp2025}, where defenders do not directly command the target but exploit an automatic reactive response of the latter. A particularly related paradigm is StringNet Herding~\cite{ChipadePanagou2019,ChipadePanagou2020,ChipadeMarellaPanagou2021}, which achieves containment through a closed formation of virtual barriers, under the assumption that the attackers avoid collisions with defenders and barriers. Three main distinctions separate our work from this literature. First, the nature of the target is different: StringNet and related indirect herding works consider \emph{external} adversarial swarms with reactive or risk-averse dynamics, whereas we consider a single \emph{internal} agent of the fleet, hijacked at the high-level command layer, whose command may be worst-case. Second, the actuation mechanism is different: our approach does not require a closed net or an impenetrable barrier, but uses the compromised target's low-level repulsive primitive as the main indirect actuation channel. Third, from an analytical standpoint, the cited literature typically provides Lyapunov stability guarantees, finite-time convergence results, or scaling-law characterizations~\cite{LamaDiBernardo2024,LamaDiBernardoKlapp2025}, whereas our contribution is a distributed dynamic-regret guarantee against a time-varying Stackelberg oracle.
}

\subsection{Contribution}
In this paper, we introduce a novel containment framework for compromised agents in multi-agent systems, with a focus on UAV and mobile-robot swarms operating in safety-critical scenarios. Unlike traditional resilient-coordination approaches centered on attack detection, isolation, or exclusion, our method exploits a structural property of many autonomous platforms: low-level collision-avoidance mechanisms typically remain active even when the high-level command layer is compromised. This property allows defender agents to indirectly influence the compromised target through repulsive-field shaping.
The main contributions are as follows:
\begin{itemize}
\item We provide an exact geometric characterization of robust one-step containment. Using support-function and normal-cone arguments, we derive equivalent containment conditions and introduce the notion of a \emph{repulsive cage}, i.e., a defender configuration that guarantees containment against all admissible adversarial commands while optionally inducing motion toward a desired destination.
\item We formulate the containment task as an online Stackelberg game between a defender team and an adversarial target. On this basis, we define a centralized Stackelberg oracle and establish existence and optimality properties of the stage problem.
\item We develop a fully distributed online approximation of the centralized oracle based only on local communication, dynamic field estimation, and neighborhood interactions. We complement the algorithm with a dynamic-regret analysis that quantifies the performance gap with respect to the centralized Stackelberg benchmark.
\end{itemize}

\section{Preliminaries}
\label{sec:preliminaries}
\subsection{Notation and Graph Theory}
We denote vectors with boldface lowercase letters and matrices with uppercase letters. We refer to the \mbox{$(i,j)$-th} entry of a matrix $A$ by $A_{ij}$. 
We represent by ${\bm 0}_n$ and ${\bm 1}_n$ vectors with $n$ {entries}, all equal to zero and to one, respectively. We use $\|\cdot\|$ for the Euclidean norm and
$\col(\cdot)$ for vertical concatenation. For a closed convex set
$C\subset\R^d$, the distance and projection of a \textcolor{black}{vector} $z\in\mathbb{R}^d$ are
$$
    \dist(z,C):=\inf_{q\in C}\|z-q\|,
    \qquad
    \proj_C(z):=\argmin_{q\in C}\|z-q\|.
$$
The normal cone of $C$ at $z\in C$ is
\begin{equation}
    N_C(z):=\{v\in\R^d:v^\top(q-z)\le 0,\ \forall q\in C\}.
    \label{eq:normal_cone}
\end{equation}
For a closed convex set $C$, its support function is
\begin{equation}
    \supp_C(v):=\sup_{q\in C} q^\top v,
    \label{eq:support_function}
\end{equation}
with the usual extended-real convention. If $C$ is compact, the supremum
is attained.
We also write $[a]_+:=\max\{a,0\}$.

Let $\mathcal{G} = (\mathcal{V}, \mathcal{E})$ be a {\em graph} with $n$ nodes $\mathcal{V}=\{ 1, \ldots, n \}$ and $e$ edges $\mathcal{E}\subseteq \mathcal{V}\times \mathcal{V}$, where $(i,j)\in \mathcal{E}$ captures the existence of a link from node $i$ to node $j\neq i$. 
A graph is said to be {\em undirected} if the existence of an edge $(i,j)\in \mathcal{E}$ implies the existence of $(j,i)\in \mathcal{E}$, while it is {\em directed} otherwise.
In this paper, we consider undirected graphs.
The neighborhood $\mathcal{N}_i$ of a node $i$ in an undirected graph is the set of nodes $j$ such that $(i,j)\in \mathcal{E}$.
The {\em degree} $d_i$ of a node $i$ in an undirected graph is the number of its edges, i.e., \mbox{$d_i = |\mathcal{N}_i|$}. The notation $\G_k=(\V,\E_k),$ defines the graph topology at time $k$. 
%
Associated with $\G_k$ is a mixing matrix
$W[k]=[w_{ij}[k]]\in\R^{n\times n}$.

\begin{assumption}
\label{ass:graph}
For every $k$, the graph $\G_k$ is connected. The matrix $W[k]$ satisfies:
\begin{align}
    &w_{ij}[k]>0
    \quad\text{only if }j=i\text{ or }j\in\N_i[k],
    \label{eq:weight_sparsity}\\
    &W[k]\mathbf{1}=\mathbf{1},
    \qquad
    \mathbf{1}^\top W[k]=\mathbf{1}^\top,
    \label{eq:doubly_stochastic}\\
    &w_{ij}[k]\ge w_{\min}>0
    \quad\text{whenever }w_{ij}[k]>0,
    \label{eq:weight_lower_bound}\\
    &\left\|W[k]-\frac{1}{n}\mathbf{1}\mathbf{1}^\top\right\|_2
    \le \theta
    \quad\text{for some }\theta\in(0,1).
    \label{eq:mixing_contraction}
\end{align}
\end{assumption}
Notice that $w_{ij}[k]$ is the weight used by node $i$ when averaging information
received from node $j$. The sparsity condition enforces
neighbor-to-neighbor communication, double stochasticity preserves network
averages, and the contraction bound is the mixing estimate used in the
regret analysis.

\subsection{Stackelberg Games}
\label{sec:prelim-stackelberg}

A Stackelberg game is a two-level hierarchical game in which one player, called the \emph{leader}, selects and commits to its action first, while the other player, called the \emph{follower}, observes such an action and responds accordingly. Unlike simultaneous games, in which players choose their actions at the same time without observing each other's moves and the relevant solution concept is typically a Nash equilibrium, the Stackelberg model features an explicit informational asymmetry: the leader \emph{commits first}, that is, it commits to its action before the follower, making it observable. The leader exploits this precedence, but must anticipate the follower's optimal response in its own optimization problem.

\begin{definition}[Stackelberg game]
\label{def:stackelberg-generale}
Let $\mathcal{U}_L$ and $\mathcal{U}_F$ denote the admissible action sets of the leader and the follower, respectively, and let
$J_L : \mathcal{U}_L\times\mathcal{U}_F\to\mathbb{R}$ be the leader's loss and
$J_F : \mathcal{U}_L\times\mathcal{U}_F\to\mathbb{R}$ be the follower's payoff.
A \emph{Stackelberg game} with leader $L$ and follower $F$ is the hierarchical problem in which the leader chooses $u_L\in\mathcal{U}_L$ while anticipating that the follower will observe $u_L$ and respond with its best response.
\end{definition}

\begin{definition}[Follower best response]
\label{def:best-response}
Given a leader action $u_L\in\mathcal{U}_L$, the follower's \emph{best-response correspondence} is the multivalued map
\begin{equation}
\mathcal{B}_F(u_L) \;:=\; \argmax_{u_F\in\mathcal{U}_F} J_F(u_L,u_F),
\label{eq:best-response-set}
\end{equation}
and every element $u_F^\star\in\mathcal{B}_F(u_L)$ is a best response of the follower to $u_L$.
\end{definition}

\begin{definition}[Stackelberg equilibrium]
\label{def:stackelberg-eq}
A pair $(u_L^\star,u_F^\star)\in\mathcal{U}_L\times\mathcal{U}_F$ is a \emph{Stackelberg equilibrium} if
\begin{equation}
u_L^\star \;\in\; \argmin_{u_L\in\mathcal{U}_L} J_L\!\bigl(u_L,\,u_F^\star(u_L)\bigr),
\qquad
u_F^\star \;\in\; \mathcal{B}_F(u_L^\star),
\label{eq:stackelberg-eq}
\end{equation}
where $u_F^\star(u_L)\in\mathcal{B}_F(u_L)$ is a selection of the best-response correspondence in Eq.~\eqref{eq:best-response-set}.
\end{definition}

When $\mathcal{B}_F(u_L)$ is multivalued, a selection rule must be specified in order to make the leader's problem well posed. Two standard conventions are the \emph{strong} or \emph{optimistic} Stackelberg convention, in which the follower selects, among its best responses, the one that is most favorable to the leader, and the \emph{weak} or \emph{pessimistic} Stackelberg convention, in which the follower selects the worst one for the leader. In what follows, we adopt a fixed selection rule consistent with a worst-case interpretation from the leader's viewpoint.

\section{Agent Model}
\label{sec:model}

Let us consider a discrete-time multi-agent system composed of $n$ regular agents, called {\em defenders} and indexed by $\V=\{1,\ldots,n\}$, and one compromised agent, called the {\em target} and indexed by $t$. Physically, this models a scenario such as a formation of UAVs in which one vehicle has been hijacked by an adversary. We denote the set of all agents by $\mathcal A:=\V\cup\{t\}$.
A key structural assumption is that the hijacking is {\em partial}, i.e., each agent operates under a two-layer control architecture. The first layer is an {\em high-level command layer} generates a nominal motion setpoint, which, for the defenders, is chosen strategically by the distributed algorithm developed in this paper (in order to accomplish the common goal); for the target, this channel has been seized by the adversary, who freely selects the target's nominal command at every time step (with the aim to deviate itself and the rest of the agents from the common goal). On the other hand, the {\em low-level safety layer} is an automatic collision-avoidance primitive embedded in the onboard hardware of every agent. It generates pairwise repulsive displacements based on inter-agent distances and cannot be disabled or overridden by any high-level input. This layer is assumed to be identical and active on both defenders and the compromised target.
The defender strategy is built on exploiting this architecture. Since the low-level repulsion of the target responds automatically to the defenders' positions, the defenders can indirectly shape the net repulsive field experienced by the target simply by choosing where to move. This field-shaping perspective represents the core idea of the methodology: the containment problem is cast as a geometric game played over induced positions rather than over commands.

Let us assume that, at discrete time $k$, the position of defender $i\in\V$ is $x_i[k]\in\R^d$ and the position of the target is $x_t[k]\in\R^d$, with $d\in\{2,3\}$. We write $x[k]=\col(x_1[k],\ldots,x_n[k])$ for the stacked defender configuration and $\mathcal A:=\V\cup\{t\}$ for the set of all agents.

The pairwise repulsive safety term acting on agent $a\in\mathcal A$ due to agent $b\in\mathcal A$, $b\ne a$, is
\begin{equation}
    \rho_{ab}(z_a,z_b)
    =
    \chi_{ab}(\|z_a-z_b\|)
    \frac{z_a-z_b}{\|z_a-z_b\|},
    \qquad z_a\ne z_b,
    \label{eq:rho_pairwise}
\end{equation}
where $\chi_{ab}(\cdot)\ge 0$ is a scalar repulsion gain that is large at short inter-agent distances and may vanish beyond a sensing radius, while the unit vector $(z_a-z_b)/\|z_a-z_b\|$ points from $b$ toward $a$, so the displacement always pushes $a$ away from $b$. The gain $\chi_{ab}$ may differ across pairs, but the functional form is the same for every interaction in the network. Each agent executes this repulsion automatically as part of its low-level layer. Notably, it computes only pairwise distances, requires no global coordination, and cannot be suppressed by the high-level command.

The information available to the defenders for their decision-making is specified in the following assumption. The set of admissible adversarial commands $\U_a\subset\R^d$ (compact and convex) bounds the high-level commands the adversary may issue, where the subscript $a$ stands for \emph{adversary}; the closed convex set $\Omega_{k+1}\subset\R^d$ is the admissible region for the target at the next time step, where the subscript $t$ will always denote the compromised target (controlled by the adversary); and $x_B\in\R^d$ is the destination toward which the defenders may wish to steer the target.
Here, the notation is intentionally asymmetric: $a$ refers to the adversary, while $t$ refers to the target agent under adversarial influence. This distinction is used consistently throughout the paper, so $u_a[k]$ is the adversary's command and $x_t[k]$ is the target's state.
Moreover, for every unit direction $v\in\R^d$, the quantities $\supp_{\U_a}(v)$ and $\supp_{\Omega_{k+1}}(v)$ denote the support-function evaluations of the sets $\U_a$ and $\Omega_{k+1}$ along $v$, namely
\[
\supp_{\U_a}(v)=\sup_{u\in\U_a}u^\top v,
\qquad
\supp_{\Omega_{k+1}}(v)=\sup_{q\in\Omega_{k+1}}q^\top v.
\]

\begin{assumption}
\label{ass:target_measurement}
At the beginning of time step $k$, each defender $i\in\V$ has access to the current target position $x_t[k]$, either through direct sensing or over the communication network, as well as to the sets $\U_a$ and $\Omega_{k+1}$, the destination $x_B$ when specified, and the repulsion laws in \eqref{eq:rho_pairwise}. For every unit direction $v$, each defender can evaluate the support values $\supp_{\U_a}(v)$ and $\supp_{\Omega_{k+1}}(v)$, either from an explicit local description of the sets or from shared mission data.
\end{assumption}

Crucially, the {\em realized} adversarial command $u_a[k]$ is unknown to the defenders at decision time: they know only the feasible set $\U_a$, not which command the adversary will execute. This forces a worst-case approach and gives rise to the Stackelberg structure formalized in Section~\ref{sec:problem_geometry}, where defenders commit first, and the adversary responds. In the rest of the section, we describe how defender high-level commands are mapped into induced defender configurations through the low-level repulsive layer, and the compromised target dynamics under adversarial commands and the aggregate repulsive field generated by the defenders.

\subsection{Defender Commands and Induced Configurations}
Let us consider that, at the beginning of time step $k$, defender $i$ is located at $x_i[k]$ and
chooses a high-level command $u_i[k]\in\R^d$. The realized motion combines
this command with the repulsive displacements generated by the embedded
low-level safety layer due to the other defenders and to the target.
Let
\begin{align}
    u[k]&:=\col(u_1[k],\ldots,u_n[k]),
    \label{eq:stacked_commands}\\
    y[k]&:=\col(y_1[k],\ldots,y_n[k])
    \label{eq:stacked_configurations}
\end{align}
denote, respectively, the command selected at time step $k$ and the
induced defender configuration generated by that command. We set
\begin{equation}
    y[k]=f_{D,k}(u[k]),
    \label{eq:induced_configuration_map}
\end{equation}
where the command-to-configuration map $f_{D,k}$ is defined componentwise
as follows: for any candidate command $u=\col(u_1,\ldots,u_n)$, where
$\rho_{ab}(\cdot,\cdot)$ is the repulsive term in Eq.~\eqref{eq:rho_pairwise}, we have
\begin{align}
    &[f_{D,k}(u)]_i
    =
    x_i[k]+u_i+h_i(x[k])+q_i(x_i[k],x_t[k]),
    \label{eq:defender_induced_position}\\
    &h_i(x[k])
    :=
    \sum_{j\in\V\setminus\{i\}}
    \rho_{ij}(x_i[k],x_j[k]),
    \label{eq:defender_defender_repulsion}\\
    &q_i(x_i[k],x_t[k])
    :=
    \rho_{it}(x_i[k],x_t[k]).
    \label{eq:defender_target_repulsion}
\end{align}
Eq.~\eqref{eq:defender_induced_position} shows that the induced
position of defender $i$ is the superposition of four terms: current
position $x_i[k]$, commanded displacement $u_i$, cumulative
defender--defender repulsion $h_i(x[k])$, and defender--target repulsion
$q_i(x_i[k],x_t[k])$. Practically, the high-level command is filtered by the low-level safety layer,
so the realized displacement is shaped by nearby agents.
Thus, $u_i[k]$ is the chosen command, and the induced position of defender
$i$ is
\begin{equation}
    y_i[k]=[f_{D,k}(u[k])]_i.
    \label{eq:induced_component}
\end{equation}
This position is generated during time step $k$ after the defender's own
low-level repulsion has acted. After execution, the
physical state is updated by $x_i[k+1]=y_i[k]$.

Let $\U_{i,k}\subset\R^d$ be the compact convex command set of defender
$i$ at time $k$, and let
\begin{equation}
    \U_{D,k}:=
    \U_{1,k}\times\cdots\times\U_{n,k}
\end{equation}
be the Cartesian product of the local command sets.
The corresponding local set of feasible induced positions is
\begin{align}
    \Y_{i,k}
    :=
    \{z_i\in\R^d:\ 
    z_i
    &=
    x_i[k]+u_i+h_i(x[k])
    \notag\\
    &\quad+
    q_i(x_i[k],x_t[k]),
    \ u_i\in\U_{i,k}\},
    \label{eq:local_Yik}
\end{align}
which defines the set of all induced positions
that defender $i$ can realize at time $k$. In particular, it is the image of $\U_{i,k}$
through the affine map
$u_i\mapsto x_i[k]+u_i+h_i(x[k])+q_i(x_i[k],x_t[k])$, where the repulsive
offset is fixed by the current geometry.
The corresponding set of feasible induced configurations is
\begin{equation}
    \Y_k
    :=
    f_{D,k}(\U_{D,k})
    =
    \{y=f_{D,k}(u):u\in\U_{D,k}\},
    \label{eq:Yk}
\end{equation}
which describes the global feasible set in induced-position
space; it collects all defender configurations generated by admissible
commands once the low-level repulsive layer has acted. This is the natural
decision space for the following Stackelberg and regret formulations, since the
target responds to induced positions rather than to raw defender commands.
Notice that the defenders choose commands $u[k]$, but the target does not react to
$u[k]$ directly: it reacts to the induced configuration
$y[k]=f_{D,k}(u[k])$. For this reason, the Stackelberg game and the
regret analysis are formulated in the induced-position variable
$y[k]\in\Y_k$.

\subsection{Compromised Target Dynamics}

For a candidate defender configuration $y\in\Y_k$, the aggregate
repulsive field induced on the target is
\begin{equation}
    g_k(y)
    :=
    \sum_{i\in\V}\rho_{ti}(x_t[k],y_i).
    \label{eq:aggregate_field}
\end{equation}
This is the low-level repulsive field experienced by the target when
the defenders occupy configuration $y$.
In particular, each summand $\rho_{ti}(x_t[k],y_i)$ is the pairwise
repulsive contribution generated by defender $i$ on the target, using the
same law defined in Eq.~\eqref{eq:rho_pairwise}; therefore, $g_k(y)$ is the
vector superposition of all defender-induced low-level effects.

The adversary controlling the target chooses the target's high-level command
$u_a[k]\in\U_a$, where $\U_a\subset\R^d$ is compact and convex. Then, the target
command set is taken time-invariant for notational clarity; the same
formulation extends directly to compact convex sets $\U_{a,k}$. The target
then evolves according to
\begin{equation}
    x_t[k+1]
    =
    x_t[k]+u_a[k]+g_k(y[k]),
    \label{eq:target_dynamics}
\end{equation}
that makes explicit the two-layer structure
for the compromised target: the adversary injects the high-level component
$u_a[k]$, while the uncompromised low-level safety layer adds the
repulsive correction $g_k(y[k])$. Hence, for fixed $x_t[k]$ and
$u_a[k]$, defender positioning acts on the target only through $g_k(y[k])$.
The defenders move to their induced configuration, and
\begin{equation}
    x[k+1]=y[k],
    \label{eq:defender_update}
\end{equation}
closes the one-step interaction: the
induced configuration selected at time $k$ becomes the defenders' physical
state at time $k+1$. Together with Eq.~\eqref{eq:target_dynamics}, this gives
the closed-loop state transition used in the containment and Stackelberg
formulations developed in the next sections.

\section{Containment Geometry}
\label{sec:problem_geometry}

This section introduces geometric conditions for guaranteeing one-step robust containment and expresses them through direction-based quantities, providing the analytical foundation that will be exploited for both centralized optimization and distributed online updates.
Let $\Omega_k\subset\R^d$ be a closed convex admissible region for the
target at time $k$. The region may be fixed or time-varying and may encode
a safe area, a moving corridor, or any other admissible target set. Given
a defender configuration $y\in\Y_k$, define the worst-case one-step
violation
\begin{equation}
    \Psi_k(y)
    :=
    \max_{u_a\in\U_a}
    \dist^2
    \left(
        x_t[k]+u_a+g_k(y),
        \Omega_{k+1}
    \right).
    \label{eq:exact_violation}
\end{equation}
Thus, $\Psi_k(y)=0$ if and only if the target remains in the admissible region
$\Omega_{k+1}$ for every admissible adversarial command. Equivalently,
since $\Omega_{k+1}$ is closed,
\begin{equation}
    \Psi_k(y)=0
    \quad\Longleftrightarrow\quad
    x_t[k]+g_k(y)+\U_a\subseteq \Omega_{k+1},
    \label{eq:exact_containment_equivalence}
\end{equation}
that expresses the one-step robust containment condition.

If the defenders aim not only to keep the target within a fixed admissible region, but also to transport it toward a destination,
$x_B\in\mathbb{R}^d$, we can define the unit vector pointing from the current target position to the destination as
\begin{equation}
    d_B[k]
    :=
    \frac{x_B-x_t[k]}{\|x_B-x_t[k]\|},
    \qquad x_t[k]\ne x_B.
    \label{eq:destination_direction}
\end{equation}
Notice that the scalar $g_k(y)^\top d_B[k]$ is the projection of the induced field onto
the direction from the current target position to the destination. Positive
values correspond to a repulsive field component that moves the target
toward $x_B$. Hence, at each time step, a desirable induced configuration
$y\in\Y_k$ has two roles: it should reduce the worst-case containment
violation $\Psi_k(y)$ and, when a destination is specified, it should make
the transport component $g_k(y)^\top d_B[k]$ positive.

Notice that the representation in Eq.~\eqref{eq:exact_violation} is exact but contains a
maximization over $u_a\in\U_a$. therefore, for analysis and algorithm design, it is
useful to rewrite the same quantity in directional form. Let
\begin{equation}
    z_k(y):=x_t[k]+g_k(y)
    \label{eq:shifted_target_center}
\end{equation}
be the target position before the adversarial command is applied. For a
unit direction $v$, define the directional containment margin
\begin{equation}
    \ell_k(y;v)
    :=
    z_k(y)^\top v
    +
    \supp_{\U_a}(v)
    -
    \supp_{\Omega_{k+1}}(v),
    \label{eq:single_normal_margin}
\end{equation}
and the worst directional margin
\begin{equation}
    c_k(y)
    :=
    \sup_{\|v\|=1}\ell_k(y;v).
    \label{eq:normal_margin}
\end{equation}
The next result shows that these directional quantities provide an exact reformulation of the one-step worst-case violation in Eq.~\eqref{eq:exact_violation}, and are therefore not merely heuristic indicators.

\begin{proposition}
\label{prop:support_containment}
For every $y\in\Y_k$,
\begin{equation}
    \Psi_k(y)=[c_k(y)]_+^2 .
    \label{eq:support_violation_formula}
\end{equation}
Consequently,
\begin{equation}
    \Psi_k(y)=0
    \quad\Longleftrightarrow\quad
    c_k(y)\le 0.
    \label{eq:support_containment_equivalence}
\end{equation}
\end{proposition}

\begin{proof}
Set $A_k(y):=z_k(y)+\U_a$. By construction, it holds
\[
\Psi_k(y)=\left(\max_{p\in A_k(y)}\dist(p,\Omega_{k+1})\right)^2.
\]
By definition, we have that for all $v \in \mathbb{R}^d$, with $\|v\|\leq 1$, for all  $p\in A_k(y)$, and all $q \in \Omega_{k+1}$,
$$
\supp_{\Omega_{k+1}}(v) \geq q^\top v.
$$
Therefore, 
$$
p^\top v - \supp_{\Omega_{k+1}}(v) \leq (p-q)^\top v\leq \|p-q\| \|v\| \leq \|p-q\|
$$
Since the above inequality holds for all $q\in \Omega_{k+1}$ we have that
$$
p^\top v - \supp_{\Omega_{k+1}}(v) \leq \inf_{q\in \Omega_{k+1}} \|p-q\| = \dist(p,\Omega_{k+1})
$$
Taking the supremum over all $v$ with $\|v\|\leq 1$ gives
\begin{equation}
    \sup_{\|v\|\leq 1}
    \left(
        p^\top v-\supp_{\Omega_{k+1}}(v)
    \right)
    \leq
    \dist(p,\Omega_{k+1}).
    \label{eq:distance_support_upper}
\end{equation}

We now prove the reverse inequality. Suppose first that
$p\notin\Omega_{k+1}$, and let
\[
    q^\star:=\proj_{\Omega_{k+1}}(p),
    \qquad
    v^\star:=
    \frac{p-q^\star}{\|p-q^\star\|}.
\]
The characterization of the Euclidean projection onto a closed convex
set gives 
\[
    (p-q^\star)^\top(q-q^\star)\leq 0,
    \qquad \forall q\in\Omega_{k+1}.
\]
Since $v^\star$ is a positive multiple of $p-q^\star$, it follows that
\[
    q^\top v^\star\leq(q^\star)^\top v^\star,
    \qquad \forall q\in\Omega_{k+1}.
\]
Therefore,
\[
    \supp_{\Omega_{k+1}}(v^\star)
    =(q^\star)^\top v^\star.
\]
Consequently,
\begin{align}
    p^\top v^\star-\supp_{\Omega_{k+1}}(v^\star)
    &=(p-q^\star)^\top v^\star\\
    &=\|p-q^\star\|\\
    &=\dist(p,\Omega_{k+1}).
\end{align}
Thus, the supremum in \eqref{eq:distance_support_upper} is at least
$\dist(p,\Omega_{k+1})$.

If $p\in\Omega_{k+1}$, then the distance is zero. Moreover,
$\supp_{\Omega_{k+1}}(v)\geq p^\top v$ for every $v$, while $v=0$
gives equality. Hence the same identity holds. We have therefore shown
that, for every $p\in\mathbb R^d$,
\begin{equation}
    \dist(p,\Omega_{k+1})
    =
    \sup_{\|v\|\leq1}
    \left(
        p^\top v-\supp_{\Omega_{k+1}}(v)
    \right).
    \label{eq:point_distance_support}
\end{equation}
Maximizing \eqref{eq:point_distance_support} over
$p\in A_k(y)$ gives
\begin{align}
    \max_{p\in A_k(y)}
    \dist(p,\Omega_{k+1})
    &=
    \sup_{p\in A_k(y)}
    \sup_{\|v\|\leq1}
    \left(
        p^\top v-\supp_{\Omega_{k+1}}(v)
    \right)                                                   \\
    &=
    \sup_{\|v\|\leq1}
    \sup_{p\in A_k(y)}
    \left(
        p^\top v-\supp_{\Omega_{k+1}}(v)
    \right)                                                   \\
    &=
    \sup_{\|v\|\leq1}
    \left(
        \supp_{A_k(y)}(v)
        -\supp_{\Omega_{k+1}}(v)
    \right).
    \label{eq:set_distance_support}
\end{align}
The order of the two suprema can be exchanged because each expression
equals the supremum of the same function over the product set
\[
A_k(y)\times
\{v\in\mathbb{R}^d:\|v\|\leq1\}.
\]

Since $A_k(y)=z_k(y)+\U_a$, the support function of $A_k(y)$ satisfies
\begin{equation}
    \supp_{A_k(y)}(v)
    =
    z_k(y)^\top v+\supp_{\U_a}(v).
    \label{eq:reachable_support}
\end{equation}
Substituting \eqref{eq:reachable_support} into
\eqref{eq:set_distance_support} and recalling the definition of
$\ell_k(y;v)$ yields
\begin{equation}
    \max_{p\in A_k(y)}
    \dist(p,\Omega_{k+1})
    =
    \sup_{\|v\|\leq1}\ell_k(y;v).
    \label{eq:ball_margin}
\end{equation}

It remains to express the supremum over the unit ball in terms of unit
directions. The choice $v=0$ gives $\ell_k(y;0)=0$. Every nonzero vector
in the unit ball can be written as
\[
    v=\alpha w,
    \qquad
    \alpha=\|v\|\in(0,1],
    \qquad
    \|w\|=1.
\]
By the positive homogeneity of support functions,
\[
    \ell_k(y;\alpha w)
    =
    \alpha\,\ell_k(y;w).
\]
Consequently, for each unit vector $w$,
\[
    \sup_{\alpha\in[0,1]}
    \ell_k(y;\alpha w)
    =
    \max\{0,\ell_k(y;w)\}.
\]
Taking the supremum over all unit directions gives
\begin{align}
    \sup_{\|v\|\leq1}\ell_k(y;v)
    &=
    \sup_{\|w\|=1}\max\{0,\ell_k(y;w)\} \\
    &=
    \max\left\{
        0,\sup_{\|w\|=1}\ell_k(y;w)
    \right\} \\
    &=[c_k(y)]_+.
    \label{eq:ball_to_sphere_margin}
\end{align}
Therefore
\begin{equation}
    \max_{p\in A_k(y)}\dist(p,\Omega_{k+1})
    =
    [c_k(y)]_+ .
\end{equation}
By squaring both sides, we obtain Eq.~\eqref{eq:support_violation_formula}, and
Eq.~\eqref{eq:support_containment_equivalence} follows immediately.
Finally, since $[c_k(y)]_+^2=0$ if and only if
$[c_k(y)]_+=0$, and the latter holds if and only if $c_k(y)\leq0$,
Eq.~\eqref{eq:support_containment_equivalence} follows.
The proof is complete. 
\color{black}
\end{proof}

The purpose of this proposition is twofold. First, it provides an equivalent scalar certificate
for robust one-step containment through $c_k(y)\le 0$. Second, it rewrites
the original worst-case term $\Psi_k(y)$ in a form directly exploitable by
the centralized and distributed updates (active-direction construction in
the next sections).

Moreover, the normal-cone interpretation clarifies the geometry of active
directions. Specifically, let us consider the special case of a fixed admissible region,
$\Omega_{k+1}=\Omega_k=\Omega$, and a target on the boundary,
$x_t[k]\in\partial\Omega$. For any outward unit normal
$v\in N_{\Omega}(x_t[k])$, where $N_{\Omega}(x_t[k])$ is the normal cone according to Eq.~\eqref{eq:normal_cone}, it holds 
\begin{equation}
    \supp_{\Omega}(v)=x_t[k]^\top v,
\end{equation}
so Eq.~\eqref{eq:single_normal_margin} becomes
\begin{equation}
    \ell_k(y;v)=g_k(y)^\top v+\supp_{\U_a}(v).
\end{equation}
Hence, along active outward normals, the defenders must generate an inward
field component that compensates the adversary's support term. This is the
bridge from the support-function characterization to the geometric notion
of a field-induced cage.
These equivalent viewpoints motivate the following definition.
\begin{definition}[Repulsive Cage]
A configuration $y\in\Y_k$ is a robust repulsive cage at time $k$ if any
of the equivalent conditions
\begin{equation}
    \Psi_k(y)=0,
    \qquad
    z_k(y)+\U_a\subseteq\Omega_{k+1},
    \qquad
    c_k(y)\le 0
\end{equation}
holds. It is a transporting repulsive cage with transport margin $m_B>0$
if, in addition,
\begin{equation}
    g_k(y)^\top d_B[k]\ge m_B.
\end{equation}
\end{definition}

Notice that $m_B$ is a design parameter
that prescribes the minimum desired component of the induced field along
the destination direction $d_B[k]$: larger values enforce more aggressive
transport toward $x_B$, while smaller values prioritize containment
feasibility.
\color{black}
In the following, this terminology is used to connect geometry and control:
the centralized Stackelberg oracle seeks configurations with small
violation and positive transport component, while the distributed
algorithm tracks those configurations online under local information and
communication constraints.

\begin{figure}[t]
\centering
\begin{tikzpicture}[scale=0.92, >=Latex, font=\small]
    \fill[blue!6] (0,0) ellipse (2.05 and 1.25);
    \draw[blue!55!black, thick] (0,0) ellipse (2.05 and 1.25);
    \node[blue!55!black] at (-1.15,-1.55) {$\Omega$};

    \coordinate (xt) at (2.05,0.12);
    \node[circle, draw=red!70!black, fill=red!55, inner sep=2.4pt,
          label={[red!70!black]above left:$x_t[k]$}] at (xt) {};

    \draw[->, thick, red!70!black] (xt) -- ++(0.95,0.08)
        node[above] {$v$};
    \draw[->, thick, orange!85!black] (xt) -- ++(0.70,-0.42)
        node[below] {$u_a$};
    \draw[->, very thick, green!45!black] (xt) -- ++(-1.15,-0.16)
        node[below left] {$g_k(y)$};

    \foreach \p/\name in {(3.25,0.95)/1,(3.32,0.25)/2,(3.08,-0.52)/3,(2.72,-1.08)/4}
    {
        \node[circle, draw=black!60, fill=black!20, inner sep=2.8pt] (d\name) at \p {};
    }

    \draw[dashed, gray!70] (d1) -- (xt);
    \draw[dashed, gray!70] (d2) -- (xt);
    \draw[dashed, gray!70] (d3) -- (xt);
    \draw[dashed, gray!70] (d4) -- (xt);
    \node[black!70] at (3.35,1.28) {defenders};

    \node[align=center, text width=3.9cm] at (-1.05,1.65)
    {$\ell_k(y;v)\le 0$\\
    blocks outward motion};
\end{tikzpicture}
\caption{Boundary geometry for a fixed admissible set. The adversary has
an outward component along $v$, while defenders on the outward side of the
target induce an inward aggregate field $g_k(y)$ through the target's
low-level safety behavior.}
\label{fig:geometry}
\end{figure}
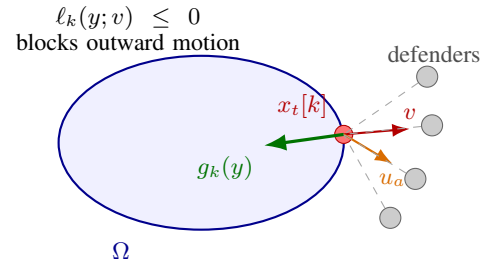

\section{Centralized Stackelberg Oracle}
\label{sec:centralized}

The previous section describes when an induced defender configuration
contains the target against all admissible adversarial commands. The goal of this section is to turn this containment requirement into a one-step Stackelberg game. At
time $k$, the defenders commit first to an induced configuration
$y\in\Y_k$. After observing this configuration, the adversary chooses the
target command $u_a\in\U_a$. In this sense, the defenders are the
Stackelberg leader and the adversary is the follower.
Notice that this sequencing assumption is operationally plausible. Once the defenders
identify that the target node has been compromised, they are the first to
start the containment procedure by selecting their induced configuration;
the compromised target then reacts through the adversarial command channel.

We first consider a centralized oracle, namely a full-information
benchmark for this stage game. The oracle has access to the full state, i.e., all variables needed by the stage problem
are centrally available at time $k$, in particular, the current defender
configuration $x[k]$, the target position $x_t[k]$, and the geometric sets
defining the one-step constraints and costs,
the feasible set $\Y_k$, the admissible target set $\Omega_{k+1}$, and
the adversarial command set $\U_a$. The oracle objective also includes a
defender-side penalty that selects among configurations with comparable
containment performance.
Notably, the oracle can be interpreted as the outcome of
perfectly shared defender intelligence and unlimited coordination. In this
paper, it is not intended as an additional physical controller, but a reference
solution used to define the performance target for the distributed scheme.

To make this objective concrete, we introduce a defender-side
regularization term that ranks configurations with similar containment
quality according to motion effort and separation preferences. Let
$\mu_x,\mu_d,\mu_t\ge 0$ be tuning weights, and let $d_{\min}>0$ and
$d_{\min}^t>0$ be desired defender--defender and defender--target
separations, respectively. For $y\in\Y_k$, we define the defender-side
regularization penalty as follows:
\begin{align}
    R_k(y)
    &=
    \mu_x\|y-x[k]\|^2
    \notag\\
    &\quad+
    \mu_d\sum_{i<j}
    [d_{\min}-\|y_i-y_j\|]_+^2
    \notag\\
    &\quad+
    \mu_t\sum_{i\in\V}
    [d_{\min}^t-\|y_i-x_t[k]\|]_+^2 .
    \label{eq:defender_penalty}
\end{align}
The first term penalizes motion from the current defender configuration, while the second term is active when two defenders are closer than
$d_{\min}$, and the third when a defender is closer than $d_{\min}^t$ to
the target. This is coherent to the fact that $R_k$ regularizes the leader's
choice: it discourages large moves and degenerate configurations.
The following regularity assumption guarantees that the Stackelberg game is
well posed.

\begin{assumption}
\label{ass:stage_regularity}
For each $k$, the sets $\U_a$ and $\Y_k$ are nonempty and compact,
$\Omega_{k+1}$ is closed and convex, $g_k$ is continuous on $\Y_k$, and
$R_k$ is continuous on $\Y_k$.
\end{assumption}

Let us now consider the adversarial response for a defender configuration
$y\in\Y_k$. We first define the follower payoff
\begin{equation}
    P_k(y,u_a)
    :=
    \dist^2
    \left(
        x_t[k]+u_a+g_k(y),
        \Omega_{k+1}
    \right),
    \label{eq:follower_payoff}
\end{equation}
which measures the one-step escape from the admissible set induced
by the command $u_a$. The follower then chooses the command that maximizes
this payoff:
\begin{equation}
    u_a^\star[k](y)
    \in
    \argmax_{u_a\in\U_a}
    P_k(y,u_a),
    \label{eq:follower_br}
\end{equation}
Accordingly, the induced worst-case violation is
\begin{equation}
    \Psi_k(y)=\max_{u_a\in\U_a}P_k(y,u_a).
\end{equation}
Notice that if several adversarial commands attain the maximum, the particular
selection of $u_a^\star[k](y)$ does not affect the centralized defender
problem; only the worst-case value $\Psi_k(y)$ is used.

The next two subsections complete the leader side of the stage game. Specifically, we define the centralized
we first define the centralized leader objective and the associated oracle minimization problem, and then we provide the corresponding
optimality condition in convex regimes.

\subsection{Centralized Defender Problem}

This subsection formulates the leader optimization problem obtained after
accounting for the follower's worst-case response.
The centralized leader anticipates the adversarial response and evaluates
each defender configuration through the reduced stage loss
\begin{equation}
    F_k(y)
    =
    \alpha\Psi_k(y)
    +
    \beta[m_B-g_k(y)^\top d_B[k]]_+^2
    +
    R_k(y),
    \label{eq:reduced_loss}
\end{equation}
where $\alpha,\beta>0$. In particular, the first term penalizes worst-case containment
violation, the second term penalizes transport deficit with respect to the
target margin $m_B$ (it is active only when $g_k(y)^\top d_B[k]<m_B$), and the third is
the defender-side penalty in Eq.~\eqref{eq:defender_penalty}.
Notice that according to Proposition \ref{prop:support_containment}, the same loss can be
written as
\begin{equation}
    F_k(y)
    =
    \alpha[c_k(y)]_+^2
    +
    \beta[m_B-g_k(y)^\top d_B[k]]_+^2
    +
    R_k(y).
    \label{eq:reduced_loss_support}
\end{equation}
This representation is the one used by the distributed update below: the
maximization over directions is handled through an active direction of
$c_k$.
Therefore, the centralized Stackelberg oracle is any minimizer
\begin{equation}
    y_k^\star
    \in
    \argmin_{y\in\Y_k} F_k(y).
    \label{eq:centralized_oracle}
\end{equation}
The following theorem formalizes the existence of a solution to the centralized Stackelberg oracle problem.
\begin{theorem}
\label{thm:oracle_existence}
Let Assumption~\ref{ass:stage_regularity} hold. For every time $k$ and every
$y\in\Y_k$, the adversarial response set
\begin{equation}
    \argmax_{u_a\in\U_a}P_k(y,u_a)
\end{equation}
is nonempty. Moreover, the centralized Stackelberg oracle set
\begin{equation}
    \argmin_{y\in\Y_k}F_k(y)
\end{equation}
is nonempty and compact.
\end{theorem}

\begin{proof}
The proof proceeds in two steps, corresponding to the two claims in the
statement.
Let us first prove the nonemptiness of the adversarial response set.
Let us fix a time $k$ and $y\in\Y_k$. The target's next position is
$x_t[k]+u_a+g_k(y)$, which is an affine (hence continuous) function of
$u_a$. The squared distance to the closed convex set $\Omega_{k+1}$ is
continuous in its argument, so the composition
\[
    u_a \;\mapsto\; P_k(y,u_a)
    = \dist^2\!\left(x_t[k]+u_a+g_k(y),\,\Omega_{k+1}\right)
\]
is continuous on $\U_a$. Since $\U_a$ is compact and $P_k(y,\cdot)$ is
continuous, the Weierstrass extreme value theorem guarantees that the
maximum is attained, i.e., $\argmax_{u_a\in\U_a}P_k(y,u_a)\neq\emptyset$.

Let us now address the nonemptiness of the oracle set.
We show that $F_k$ is continuous on the compact set $\Y_k$, so that the
Weierstrass theorem applies again.
Let us analyze the continuity of each term in $F_k$. 
 The joint map
    $(y,u_a)\mapsto P_k(y,u_a)$ is continuous on the compact set
    $\Y_k\times\U_a$, since $g_k$ is continuous on $\Y_k$ by
    Assumption~\ref{ass:stage_regularity} and the remaining operations
    (affine map, squared distance) are continuous. By Berge's maximum
    theorem, the value function
    $\Psi_k(y)=\max_{u_a\in\U_a}P_k(y,u_a)$
    is therefore continuous on $\Y_k$.
For the transport term, we have that the inner product $g_k(y)^\top d_B[k]$ is continuous in $y$ since
    $g_k$ is continuous. The map $s\mapsto[m_B-s]_+^2$ is continuous,
    so $y\mapsto[m_B-g_k(y)^\top d_B[k]]_+^2$ is continuous on $\Y_k$. Finally, the continuity of $R_k$ holds directly by
    Assumption~\ref{ass:stage_regularity}.
Since $F_k = \alpha\Psi_k + \beta[\cdots]_+^2 + R_k$ is a sum of
continuous functions scaled by positive constants, $F_k$ is continuous on
$\Y_k$. As $\Y_k$ is compact, the Weierstrass theorem guarantees that
$F_k$ attains its minimum, so $\argmin_{y\in\Y_k}F_k(y)\neq\emptyset$.
  Finally, compactness follows because the oracle set is a closed subset of
  the compact set $\Y_k$: it equals $F_k^{-1}(\{m_k\})\cap\Y_k$, where
  $m_k:=\min_{y\in\Y_k}F_k(y)$, and the continuity of $F_k$ ensures that
  the level set is closed. This completes our proof.
\end{proof}

Notice that Theorem~\ref{thm:oracle_existence} does not require uniqueness of the
oracle. In general, several configurations may attain the same minimum
value of $F_k$. The compactness of
$\argmin_{y\in\Y_k}F_k(y)$ guarantees that it is always possible to select at least
one representative $y_k^\star$ (for instance, via a deterministic
tie-breaking rule). This is useful in the online analysis, where the
benchmark sequence $\{y_k^\star\}$ and its path variation are defined.

\subsection{First-Order Characterization}

This subsection derives the first-order optimality condition associated
with the centralized oracle problem.
Notice that the exact field-shaping problem may be nonconvex because of the repulsive
geometry and the feasible induced-configuration set. Still, it is useful
to record the stationarity condition that an oracle solution satisfies in
convex regimes: the local decrease directions of the reduced loss must be
balanced by the normal directions generated by the feasible set. In particular, the
following proposition states this condition.
Let us first introduce the following assumption.
\begin{assumption}
\label{ass:convex_oracle}
Fix a time instant $k$. The set $\Y_k$ is nonempty, closed, and convex, and
$F_k$ is convex and subdifferentiable on a neighborhood of $\Y_k$.
\end{assumption}
We are now in position to formalize the convex oracle optimality.
\begin{proposition}
\label{prop:oracle_optimality}
Under Assumption~\ref{ass:convex_oracle}, a point
$y_k^\star\in\Y_k$ solves Eq.~\eqref{eq:centralized_oracle} if and only if
\begin{equation}
    0\in \partial F_k(y_k^\star)+N_{\Y_k}(y_k^\star).
    \label{eq:oracle_optimality}
\end{equation}
\end{proposition}

\begin{proof}
Let us consider the time index $k$ and abbreviate $\Y:=\Y_k$ and $F:=F_k$ to improve readability.
Let us also consider the indicator function of the feasible set,
\[
    I_{\Y}(y)
    =
    \begin{cases}
        0, & y\in\Y,\\
        +\infty, & y\notin\Y.
    \end{cases}
\]
Then, minimizing $F$ over $\Y$ is equivalent to minimizing the extended-value
function $F+I_{\Y}$ over the whole space $\R^d$. In other words,
$y_k^\star$ solves Eq.~\eqref{eq:centralized_oracle} if and only if it is a
global minimizer of $F+I_{\Y}$.
Since $F$ is convex and subdifferentiable on a neighborhood of $\Y$, the
standard convex Fermat rule applies to the unconstrained problem and yields
\[
    y_k^\star\text{ is optimal}
    \quad\Longleftrightarrow\quad
    0\in\partial(F+I_{\Y})(y_k^\star).
\]
At this point we need to compute the subdifferential of the sum. By the convex
subdifferential sum rule,
\[
    \partial(F+I_{\Y})(y)
    =
    \partial F(y)+\partial I_{\Y}(y).
\]
The subdifferential of the indicator is exactly the normal cone,
$\partial I_{\Y}(y)=N_{\Y}(y)$, where $N_{\Y}$ is defined in
Eq.~\eqref{eq:normal_cone}. Therefore,
\[
    0\in\partial(F+I_{\Y})(y_k^\star)
    \quad\Longleftrightarrow\quad
    0\in\partial F(y_k^\star)+N_{\Y}(y_k^\star),
\]
which is the claimed condition. The proof is complete.
\end{proof}

Notice that is $F_k$ is differentiable, \eqref{eq:oracle_optimality} reduces to
\begin{equation}
    -\nabla F_k(y_k^\star)\in N_{\Y_k}(y_k^\star).
\end{equation}
Thus the negative descent direction must be balanced by the feasible-set
geometry. If $y_k^\star$ lies in the relative interior of $\Y_k$, this
condition reduces further to $\nabla F_k(y_k^\star)=0$.

\begin{remark}
For the original nonconvex field-shaping problem, the same expression is
best read as a surrogate or local stationarity condition, with convex
subdifferentials replaced by the appropriate nonsmooth generalized
differential when needed.
\end{remark}

\section{Distributed Online Implementation}
\label{sec:distributed}
\color{black}
The centralized oracle in Eq.~\eqref{eq:centralized_oracle} provides the
stage benchmark $y_k^\star$, but its computation requires global information
about the defender configuration and the aggregate field $g_k(y)$. The
goal of this section is to build a fully distributed online controller
that approximates this benchmark using only local sensing and
distributed communication.

Specifically, let us consider a set of $n$ regular agents (defenders) intaracting according to an undirected graph topology $\mathcal{G}=(\mathcal{V}, \mathcal{E})$. At time step $k$, defender $i$ computes its command using only locally
available data, i.e., its own state and command history, target measurements
available under Assumption~\ref{ass:target_measurement}, and messages
received from neighbors in $\N_i[\tau]$ for $\tau\le k$. In particular,
defenders know the current target position and the admissible adversarial
set, but they do not know the realized adversarial command at decision
time. The policy is therefore causal and distributed when each $u_i[k]$ is a
function of this local information only.
The communication model from Assumption \ref{ass:graph} is used to
reconstruct global field information in a distributed fashion. The mixing
weights $w_{ij}[k]$ determine how defender $i$ combines neighbor data at
time $k$, and they enable online tracking of the aggregate repulsive field
without centralized aggregation. The resulting controller is built from the
support-based loss representation in Eq.~\eqref{eq:reduced_loss_support}:
at each stage, defenders exploit a direction-dependent surrogate of the robust
containment term and treat the mismatch with the exact worst-case direction
as a controllable approximation error.
To evaluate this online approximation, we adopt a regret viewpoint. In particular,
regret measures the cumulative loss gap between the distributed decisions
and a reference benchmark, and in our case the reference is the
time-varying centralized oracle sequence.
Before introducing the full update, we decompose the construction into three
steps. We first define the online performance metrics and the target
problem and introduce the distributed field-tracking recursion.
Then, we combine the previous parts into the active-direction
distributed update used by the defenders.

\subsection{Online Performance Criterion}

This subsection formalizes the online objective and the performance metrics
used to evaluate causal distributed containment policies.
As standard in online optimization, regret measures the cumulative loss
gap between the realized online decisions and a reference benchmark.
The benchmark is time-varying since the stage-wise oracle sequence is not necessarily stationary (i.e., the optimal configuration may change over time, leading to a nonzero path variation of
$\{y_k^\star\}$), so we use dynamic regret rather than static regret.
Let $y[k]$ be the induced configuration generated by a causal distributed
policy and let $y_k^\star$ be the centralized oracle in
Eq.~\eqref{eq:centralized_oracle}. The dynamic regret over horizon $T$ is
\begin{equation}
    \operatorname{Reg}_T^{\operatorname{dyn}}
    :=
    \sum_{k=1}^T F_k(y[k])
    -
    \sum_{k=1}^T F_k(y_k^\star).
    \label{eq:dynamic_regret}
\end{equation}
Notice that this term measures the total accumulated loss difference over the horizon: the first sum is the
cumulative cost of the actual distributed trajectory, while the second sum is the
cumulative cost of the oracle trajectory. If dynamic regret is sublinear in $T$, the
average regret per step vanishes, meaning the distributed policy asymptotically
performs as well as the centralized oracle.
Moreover, the path variation of the oracle sequence is
\begin{equation}
    P_T
    :=
    \sum_{k=1}^{T-1}\|y_{k+1}^\star-y_k^\star\|,
    \label{eq:path_variation}
\end{equation}
and quantifies how much the oracle configuration moves from one step to the next.
If the oracle is stationary, i.e., $P_T$ is bounded, tracking is easier; if the oracle drifts significantly,
the benchmark itself becomes harder to follow, and a larger regret is unavoidable even
for a perfect tracker.
We can also define the cumulative containment violation as
\begin{equation}
    V_T^\Psi
    :=
    \sum_{k=1}^T\Psi_k(y[k]), 
    \label{eq:cumulative_violation}
\end{equation}
where each term $\Psi_k(y[k])$ is zero when the
target is robustly contained, and positive when the defenders fail to prevent escape despite
their best effort. Sublinearity in $V_T^\Psi$ ensures that constraint violations become rare
on average as the horizon grows.
We can now formally introduce the central problem addressed in this work.
\begin{problem}[Online Distributed Stackelberg Containment]
\label{prob:main}
Given a horizon $T$, admissible regions $\{\Omega_k\}_{k=0}^T$, a
destination $x_B$, and a communication graph sequence satisfying
Assumption \ref{ass:graph}, design causal distributed defender policies
under the information model of Assumption \ref{ass:target_measurement}.
The policies generate feasible induced configurations
$y[k]\in\Y_k$ and are evaluated by the dynamic regret
$\operatorname{Reg}_T^{\operatorname{dyn}}$ in
\eqref{eq:dynamic_regret} and the cumulative violation $V_T^\Psi$ in
\eqref{eq:cumulative_violation}. A policy achieves sublinear online
Stackelberg containment on a class of problem sequences if
\begin{equation}
    \operatorname{Reg}_T^{\operatorname{dyn}}=o(T),
    \qquad
    V_T^\Psi=o(T).
\end{equation}
\end{problem}

Notably, this problem formulation couples two objectives, i.e., achieving sublinear dynamic regret (optimization efficiency)
while maintaining sublinear cumulative violations (physical safety). This dual requirement is non-trivial,
as low regret alone does not guarantee constraint satisfaction if the benchmark oracle itself violates the
admissible region. The next subsections develop a distributed algorithm that balances these competing goals.

\subsection{Tracking the Aggregate Field}
To achieve distributed coordination without centralized gradient computation, each defender
must estimate and track the aggregate field $g_k(y)$ defined in Eq.~\eqref{eq:aggregate_field},
using only local information and peer-to-peer communication. This subsection develops the
consensus-based mechanism that enables this distributed tracking.
Let us first consider that the aggregate field $g_k(y)=\sum_{i\in\V}\rho_{ti}(x_t[k],y_i)$ is a sum
of local contributions. Then, we can define the contribution of defendeer $i$ as
\begin{equation}
    r_i[k]:=\rho_{ti}(x_t[k],y_i[k]).
    \label{eq:local_contribution}
\end{equation}
that expresses the repulsive effect of defender $i$ on the target at time $k$.
This is the fundamental building block, since the aggregate field is the sum of all such
pairwise contributions, each defender only needs to compute its own local term.
Moreover, each defender maintains an estimate $s_i[k]\in\R^d$ of the average of
these contributions over the defender network, namely
$(1/n)\sum_{j=1}^n r_j[k]$. The dynamic average-consensus
recursion is
\begin{equation}
    s_i[k+1]
    =
    \sum_{j=1}^n w_{ij}[k]s_j[k]
    +
    r_i[k+1]-r_i[k],
    \label{eq:dynamic_average_consensus}
\end{equation}
through which agents diffuse information across the network. In particular, each defender mixes estimates from
neighbors and adds the time-varying innovation $r_i[k+1]-r_i[k]$. The doubly
stochastic weights $w_{ij}[k]$ ensure that consensus is achieved while respecting the
communication graph topology.
To this aim, the initialization is consistent if the average of the estimator variables
matches the average of the initial local contributions:
\begin{equation}
    \frac{1}{n}\sum_{i=1}^n s_i[0]
    =
    \frac{1}{n}\sum_{i=1}^n r_i[0].
    \label{eq:consistent_initialization}
\end{equation}
A consistent initialization ensures that the consensus dynamics start with the correct average;
otherwise, all subsequent estimates would be biased. The simple choice $s_i[0]=r_i[0]$ satisfies
this condition and is used in practice.
Since
$W[k]$ is doubly stochastic, \eqref{eq:dynamic_average_consensus}
preserves the correct arithmetic average over time:
\begin{equation}
    \frac{1}{n}\sum_{i=1}^n s_i[k]
    =
    \frac{1}{n}\sum_{i=1}^n r_i[k].
    \label{eq:average_preservation}
\end{equation}
Notice that this averaging property is the key invariant. By construction of the doubly stochastic weights,
the collective average of all estimators remains equal to the true aggregate average at all times,
independent of local estimate deviations. Consequently, when consensus disagreement among
$s_i[k]$ is sufficiently small, each defender's local estimate $ns_i[k]$ accurately approximates
the entire aggregate field $g_k(y[k])$.
The practical implementation is clarified in the following remark.
\begin{remark}[Network Size Knowledge]
The implementation in \eqref{eq:local_aggregate_estimate} requires each
defender to know the network size $n$: the consensus state $s_i[k]$
tracks an average, and the factor $n$ converts that average into the
aggregate field. If $n$ is not known a priori, it can be estimated by a
distributed network-size estimation routine before running the controller, or an equivalent sum-tracking recursion can be used whose local states estimate $g_k(y[k])$ directly. This design choice does not affect the main regret bounds.
\end{remark}
In practice, the consensus estimates $s_i[k]$ do not instantaneously track the time-varying
true average $(1/n)\sum_j r_j[k]$ due to two effects: (i) time variation in the target position
and defender configuration, which change $r_i[k]$, and (ii) finite convergence rate of the consensus
algorithm. These tracking errors are quantified formally below and propagated through the regret
analysis.
To state the tracking property, let us define the estimator disagreement as
\begin{equation}
    D_s[k]
    :=
    \left(
        \sum_{i=1}^n
        \|s_i[k]-\bar s[k]\|^2
    \right)^{1/2},
    \label{eq:estimator_disagreement}
\end{equation}
which measures how far individual estimates deviate from their collective average, by considering
\[
    \bar s[k]:=\frac{1}{n}\sum_{i=1}^n s_i[k].
\]
For $k\ge 1$, let us also define the contribution variation disagreement as
\begin{equation}
    D_r[k]
    :=
    \left(
        \sum_{i=1}^n
        \|\Delta r_i[k]-\overline{\Delta r}[k]\|^2
    \right)^{1/2},
    \label{eq:contribution_variation_disagreement}
\end{equation}
where
\[
    \Delta r_i[k]:=r_i[k]-r_i[k-1],
    \qquad
    \overline{\Delta r}[k]
    :=
    \frac{1}{n}\sum_{i=1}^n\Delta r_i[k].
\]
Thus, $D_s[k]$ measures the disagreement among the local estimates, while
$D_r[k]$ measures how nonuniformly the local field contributions change.
For a horizon $T$, let us also define the cumulative field-variation budget as
\begin{equation}
    V_T^r:=
    \sum_{\tau=1}^T
    \left(
        \sum_{i=1}^n
        \|r_i[\tau]-r_i[\tau-1]\|^2
    \right)^{1/2}.
    \label{eq:field_variation}
\end{equation}
This quantity captures the total path length of the aggregate field trajectory, as it accumulates
how rapidly each defender's local contribution changes over the horizon. Notice that a slowly-varying field
$V_T^r$ is easier to track and induces smaller consensus errors, whereas rapid field
changes require faster consensus convergence to maintain accuracy.

We are now in position to state the following result that quantifies how fast the consensus disagreement decays despite
time variation in the field.
\begin{theorem}
\label{prop:field_tracking}
Under Assumption \ref{ass:graph}, the dynamic average-consensus recursion
in Eq.~\eqref{eq:dynamic_average_consensus} satisfies the following bound for
every $k\ge 1$
\begin{align}
    D_s[k]
    &\le
    \theta^k
    D_s[0]+
    \sum_{\tau=1}^k
    \theta^{k-\tau}
    D_r[\tau],
    \label{eq:tracking_bound}
\end{align}
where $\theta$ is the mixing contraction constant in
Eq.~\eqref{eq:mixing_contraction}.
Consequently, for every horizon $T$,
\begin{equation}
    \sum_{k=1}^T
    D_s[k]
    \le
    \frac{
        D_s[0]+V_T^r
    }{1-\theta}.
    \label{eq:cumulative_tracking_bound}
\end{equation}
\end{theorem}

\begin{proof}
To prove our statement, we need to track how consensus disagreement evolves over time using a Lyapunov-type approach. 
Let us define first the network-averaging and centering projections $J:=(1/n)\mathbf{1}\mathbf{1}^\top$ and $\Pi:=I_n-J$,
and the stacked vectors $r[k]:=\col(r_1[k],\ldots,r_n[k])$ and $s[k]:=\col(s_1[k],\ldots,s_n[k])$.
The projections $\widetilde s[k]:=(\Pi\otimes I_d)s[k]$ and $\Delta r[k]:=r[k]-r[k-1]$ extract deviations
from the network average: the blocks of $\widetilde s[k]$ are $s_i[k]-\bar s[k]$, and the blocks of
$(\Pi\otimes I_d)\Delta r[k]$ are $\Delta r_i[k]-\overline{\Delta r}[k]$.
By definition of $D_s[k]$ and $D_r[k]$, we have $\|\widetilde s[k]\|_2=D_s[k]$ and
$\|(\Pi\otimes I_d)\Delta r[k]\|_2=D_r[k]$.

Since $W[k]$ is doubly stochastic, we have $J W[k] = J$, which implies $\Pi W[k] = (W[k]-J)\Pi$.
Applying the projection $\Pi\otimes I_d$ to the consensus recursion in Eq.~\eqref{eq:dynamic_average_consensus} we have
\begin{equation}
    \widetilde s[k+1]
    =
    ((W[k]-J)\otimes I_d)\widetilde s[k]
    +
    (\Pi\otimes I_d)(r[k+1]-r[k]).
\end{equation}
Taking the $\ell_2$ norm and applying the triangle inequality, together with the mixing contraction
bound in Eq.~\eqref{eq:mixing_contraction}, we obtain the one-step estimate
\begin{equation}
    D_s[k+1]
    \le
    \theta D_s[k]+D_r[k+1].
    \label{eq:one_step_tracking}
\end{equation}
Applying this inequality recursively from $k=0$ to $k-1$ yields
\begin{equation}
    D_s[k]
    \le
    \theta^kD_s[0]
    +
    \sum_{\tau=1}^k\theta^{k-\tau}D_r[\tau],
\end{equation}
which corresponds to Eq.~\eqref{eq:tracking_bound}. To obtain the cumulative bound, we sum this from $k=1$ to $T$, then we have
\begin{equation*}
    \sum_{k=1}^T D_s[k]
    \le
    D_s[0]\sum_{k=1}^T\theta^k
    +
    \sum_{k=1}^T\sum_{\tau=1}^k\theta^{k-\tau}D_r[\tau].
\end{equation*}
By exchanging the order of summation in the double sum it holds
\begin{equation*}
    \sum_{k=1}^T D_s[k]
    \le
    D_s[0]\sum_{k=1}^T\theta^k
    +
    \sum_{\tau=1}^T D_r[\tau]\sum_{k=\tau}^T\theta^{k-\tau}.
\end{equation*}
Since $\theta \in [0,1)$, the geometric series bounds give $\sum_{k=1}^T\theta^k < \frac{1}{1-\theta}$
and $\sum_{k=\tau}^T\theta^{k-\tau} \le \frac{1}{1-\theta}$. Therefore,
\begin{equation*}
    \sum_{k=1}^T D_s[k]
    \le
    \frac{D_s[0]}{1-\theta}
    +
    \frac{1}{1-\theta}
    \sum_{\tau=1}^T D_r[\tau].
\end{equation*}
Finally, since $D_r[\tau] \le \|\Delta r[\tau]\|$ (projection reduces norm), we have
$\sum_{\tau=1}^T D_r[\tau] \le V_T^r$ by definition of Eq.~\eqref{eq:field_variation}.
This concludes our proof.
\end{proof}
Notice that this result establishes that cumulative consensus disagreement grows sublinearly in the horizon
if the field itself is slowly-varying. In particular, the bound decays exponentially with the network mixing rate
$\theta$. The dependency on $D_s[0]$ (initial disagreement) is also exponentially
suppressed, so even a poorly initialized consensus can recover. Crucially, the bound is
constructive, in that it directly bounds the tracking error in terms of quantities that will
be controlled in the regret analysis. In particular, $V_T^r$ arises naturally as the total variation
of the aggregate field, and both terms in the cumulative bound scale as $O(D_s[0] + V_T^r)/(1-\theta)$,
which will be propagated into sublinear regret bounds provided $V_T^r = o(T)$.

\subsection{Active-Direction Distributed Update}

This subsection turns the centralized objective into a distributed online
update that each defender can execute locally. The main issue is that
$F_k$ is centralized, as $c_k(y)$ depends on the full
aggregate field and includes a maximization over all directions. The
proposed recursion addresses these two points separately: a dynamic consensus algorithm
is exploited to estimate the aggregate field, and an active-direction mechanism
approximates the directional maximization.
To this aim, let us define the active-direction set as
\begin{equation}
    \mathcal D_k(y)
    :=
    \argmax_{\|v\|=1}\ell_k(y;v).
    \label{eq:active_direction_set}
\end{equation}
In particular, when $\nu_k\in\mathcal D_k(y[k])$, the vector $\nu_k$ selects the
supporting halfspace of $\Omega_{k+1}$ corresponding to the largest current
containment margin. In this sense, $\nu_k$ is the most critical direction
for one-step robustness at $y[k]$. Therefore, a step based on this single
direction is a valid subgradient step for the max term $c_k(y)$.
A defender $i$ forms first the aggregate-field estimate as
\begin{equation}
    \widehat g_i[k]:=n s_i[k],
    \label{eq:local_aggregate_estimate}
\end{equation}
and then, using this estimate, it evaluates the local directional margin
\begin{equation}
    \widehat\ell_{i,k}(v)
    :=
    (x_t[k]+\widehat g_i[k])^\top v
    +
    \supp_{\U_a}(v)
    -
    \supp_{\Omega_{k+1}}(v).
    \label{eq:local_directional_margin}
\end{equation}
In the last equation, the support-function terms are
common mission information. The only network-dependent quantity is thus
$\widehat g_i[k]$, i.e., the locally reconstructed aggregate field.

In order to obtain a common $\nu_k$, the defenders run a short consensus routine on a
finite direction dictionary. Specifically, let
\begin{equation}
    \mathcal S_\delta:=\{v^1,\ldots,v^M\}\subset\{v\in\R^d:\|v\|=1\}
\end{equation}
be a shared $\delta$-net of the unit sphere, meaning that for every unit
vector $v$ there exists some $v^m\in\mathcal S_\delta$ with
$\|v-v^m\|\le\delta$. Defender $i$ initializes the score vector
$a_i^0[k]\in\R^M$ by
\begin{equation}
    [a_i^0[k]]_m:=\widehat\ell_{i,k}(v^m),
    \qquad m=1,\ldots,M,
    \label{eq:direction_score_initialization}
\end{equation}
and then performs $L_\nu$ communication rounds
\begin{equation}
    a_i^{\ell+1}[k]
    =
    \sum_{j=1}^n w_{ij}^{\nu}[\ell,k]a_j^\ell[k],
    \qquad
    \ell=0,\ldots,L_\nu-1,
    \label{eq:direction_score_consensus}
\end{equation}
where $w_{ij}^{\nu}[\ell,k]$ respects the communication graph used during
the inner agreement routine. Let us define the average score vector as
\begin{equation}
    \bar a[k]:=\frac{1}{n}\sum_{i=1}^n a_i^0[k].
\end{equation}
If the field tracker is consistently initialized, then
$[\bar a[k]]_m=\ell_k(y[k];v^m)$ for each $m$. Hence, averaging the score
vectors gives the directional scores that a centralized controller would
compute on the dictionary.
After the inner rounds, each defender selects
\begin{equation}
    m_i[k]\in\argmax_{m\in\{1,\ldots,M\}}[a_i^{L_\nu}[k]]_m
    \label{eq:direction_index_selection}
\end{equation}
with a deterministic tie-breaking rule. A final max-consensus on the pairs
\[
    \left([a_i^{L_\nu}[k]]_{m_i[k]},-m_i[k]\right)
\]
enforces a common index after at most the graph diameter communication
rounds on a fixed connected graph, or over an interval whose union graph is
connected in the time-varying case. In particular, the second component selects the
smallest index among equal scores. The defenders then set
\begin{equation}
    \nu_k:=v^{m_k}.
\end{equation}
With finite score-consensus rounds, approximation and communication effects
appear in the active-direction component of $e_k$. In particular, if
$\ell_k(y[k];v)$ is $L_v$-Lipschitz in $v$ and
\begin{equation}
    \epsilon_{\nu,k}
    :=
    \max_i
    \|a_i^{L_\nu}[k]-\bar a[k]\|_\infty,
\end{equation}
then the selected direction satisfies the representative estimate
\begin{equation}
    c_k(y[k])-\ell_k(y[k];\nu_k)
    \le
    L_v\delta+2\epsilon_{\nu,k},
    \label{eq:direction_approximation_error}
\end{equation}
whenever the defenders agree on $\nu_k$. This makes explicit the two
relevant error sources, i.e., the dictionary resolution ($\delta$) and residual
consensus disagreement ($\epsilon_{\nu,k}$).

Notice that the ideal case used in the regret analysis is recovered when the agreed
direction belongs to $\mathcal D_k(y[k])$. The previous construction
quantifies the gap between this ideal active direction and the one produced
with finite communication.
When lower computational overhead is preferred, a cheaper geometric rule is
also possible. To this aim, let us define the local predicted target position as
\begin{equation}
    \widehat p_i[k]:=x_t[k]+\widehat g_i[k].
    \label{eq:local_predicted_target}
\end{equation}
If $\widehat p_i[k]\notin\Omega_{k+1}$, the direction
\begin{equation}
    \nu_k
    =
    \frac{
        \widehat p_i[k]-\proj_{\Omega_{k+1}}(\widehat p_i[k])
    }{
        \|\widehat p_i[k]-\proj_{\Omega_{k+1}}(\widehat p_i[k])\|
    }
    \label{eq:active_direction_projection}
\end{equation}
points from the closest admissible point toward the predicted violation.
This rule is computationally attractive because it avoids dictionary
scoring, but oracle-level regret guarantees require its cumulative
approximation error to be sublinear, as formalized in Section~\ref{sec:regret}.
Therefore, given the common direction $\nu_k$, defender $i$ computes
\begin{align}
    \widehat c_{i,k}
    &:=
    \widehat\ell_{i,k}(\nu_k),
    \label{eq:local_margin_estimate}\\
    \widehat \tau_{i,k}
    &:=
    m_B-\widehat g_i[k]^\top d_B[k].
    \label{eq:local_transport_deficit}
\end{align}
Hence, $\widehat c_{i,k}$ estimates the active containment margin, whereas
$\widehat \tau_{i,k}$ estimates the transport deficit. Let
\begin{equation}
    A_i[k]
    :=
    \nabla_{y_i}\rho_{ti}(x_t[k],y_i[k])
    \in\R^{d\times d}
    \label{eq:local_jacobian}
\end{equation}
be the local Jacobian of defender $i$'s contribution to the target field.
If $\nu_k\in\mathcal D_k(y[k])$ and exact aggregate information is
available, the $i$-th block of a subgradient of
\eqref{eq:reduced_loss_support} at $y[k]$ is
\begin{align}
    \xi_{i,k}
    &=
    \nabla_{y_i}R_k(y[k])
    -
    2\beta[\tau_k(y[k])]_+ A_i[k]^\top d_B[k]
    \notag\\
    &\quad+
    2\alpha[c_k(y[k])]_+ A_i[k]^\top \nu_k,
    \label{eq:centralized_active_subgradient}
\end{align}
where
\begin{equation}
    \tau_k(y):=m_B-g_k(y)^\top d_B[k].
\end{equation}
This follows from the max representation of $c_k$ and Danskin's theorem.
At this point, the distributed update replaces the centralized quantities in
Eq.~\eqref{eq:centralized_active_subgradient} with local estimates.
Let $R_{i,k}$ denote defender $i$'s local contribution to the penalty
$R_k$, built from its effort term and the separation penalties computable
from local measurements and neighbor messages. The local weight of the
containment term is
\begin{equation}
    \gamma_i[k]:=2\alpha[\widehat c_{i,k}]_+ .
    \label{eq:local_containment_weight}
\end{equation}
A concrete distributed descent direction is then
\begin{align}
    d_i[k]
    &:=
    \nabla R_{i,k}(y_i[k])
    -
    2\beta[\widehat \tau_{i,k}]_+ A_i[k]^\top d_B[k]
    \notag\\
    &\quad+
    \gamma_i[k] A_i[k]^\top \nu_k.
    \label{eq:concrete_direction}
\end{align}
In the active-direction regime and with exact aggregate-field estimates,
Eq.~\eqref{eq:concrete_direction} is exactly the $i$-th block of a subgradient
of the centralized loss in Eq.~\eqref{eq:reduced_loss_support}. With local
estimates, it becomes a perturbed subgradient.

Finally, the induced-configuration update is
\begin{equation}
    y_i[k+1]
    =
    \proj_{\Y_{i,k}}
    \left[
        y_i[k]-\eta_y d_i[k]
    \right].
    \label{eq:distributed_primal}
\end{equation}
In order to be able to implement this induced-position step physically, defender $i$ must map
the selected induced position to an executable command. Let us define the known
offset as
\begin{equation}
    b_i[k]
    :=
    x_i[k]+h_i(x[k])+q_i(x_i[k],x_t[k]).
    \label{eq:command_offset}
\end{equation}
Then Eq.~\eqref{eq:defender_induced_position} gives
$[f_{D,k}(u)]_i=b_i[k]+u_i$. Hence, if the projected point
$y_i[k+1]\in\Y_{i,k}$ is attainable, defender $i$ applies
\begin{equation}
    u_i[k]=y_i[k+1]-b_i[k].
    \label{eq:exact_command_recovery}
\end{equation}
If numerical errors, model mismatch, or actuator constraints make exact
attainment impossible, defender $i$ can exploit the least-squares command
\begin{equation}
    u_i[k]
    =
    \proj_{\U_{i,k}}
    \left(y_i[k+1]-b_i[k]\right),
    \label{eq:approx_command_recovery}
\end{equation}
which equivalently solves
\begin{equation}
    \min_{v_i\in\U_{i,k}}
    \|b_i[k]+v_i-y_i[k+1]\|^2 .
    \label{eq:command_recovery_least_squares}
\end{equation}
Notice that the resulting residual is an actuation error and can be absorbed into the
distributed error term used in the regret analysis.

To conclude the section, let us clarify the timing convention of the distributed recursion. At time step $k$, defender $i$ uses
$y_i[k]$, $s_i[k]$, the current target measurement, the selected direction
$\nu_k$, and neighbor messages to compute $d_i[k]$ and $y_i[k+1]$. It then
recovers and applies $u_i[k]$ through
Eq.~\eqref{eq:exact_command_recovery} or
Eq.~\eqref{eq:approx_command_recovery}. After the move, defender $i$ evaluates
$r_i[k+1]$ and updates $s_i[k+1]$ through
Eq.~\eqref{eq:dynamic_average_consensus}. By collecting these steps, we define the
distributed recursion as an explicit online implementation.

\section{Regret Analysis}
\label{sec:regret}

This section quantifies the performance of the distributed recursion in tracking the
centralized Stackelberg benchmark over time. The recursion introduced in
Section~\ref{sec:distributed} generates a feasible induced-configuration
sequence $y[1],y[2],\ldots$, which we compare with the oracle sequence
$y_1^\star,y_2^\star,\ldots$ through the dynamic regret
in Eq.~\eqref{eq:dynamic_regret}. The goal is to identify conditions under which
the average regret vanishes and to make explicit how network and
active-direction inaccuracies contribute to the final bound.
In particular, the analysis exploits an online projected-gradient view of the distributed
controller. At time $k$, the ideal active-direction step would use a
centralized subgradient $\xi_k\in\partial F_k(y[k])$ generated by an active
direction $\nu_k\in\mathcal D_k(y[k])$. The actual distributed update uses
local estimates of the aggregate field, a direction obtained by the
score-consensus routine, and a recovered physical command. The difference
between the ideal subgradient and the aggregate direction implemented by
the defenders is collected in an error term $e_k$.

Some temporal regularity of the benchmark is necessary to ensure tracking performance. In particular, if the optimizer
$y_k^\star$ can move an order-one distance at every time step, no causal
online method can track it with vanishing average loss. The path variation
$P_T$ in Eq.~\eqref{eq:path_variation} measures this motion and appears
explicitly in the regret bound.
To formalize this projected-gradient view, we collect the required
regularity and error assumptions in the next statement.
\begin{assumption}
\label{ass:regret}
There exists a nonempty compact convex set $\Y$ of diameter $D$ such that
$\Y_k\subseteq \Y$ for all $k$, where
\begin{equation}
    D:=\sup\{\|y-z\|:y,z\in\Y\}.
    \label{eq:diameter_Y}
\end{equation}
For each $k$, the loss $F_k$ in
Eq.~\eqref{eq:reduced_loss_support} is convex on $\Y$. Let
$\xi_k\in\partial F_k(y[k])$ denote a centralized subgradient obtained from
an active direction $\nu_k\in\mathcal D_k(y[k])$. The aggregate
effect of the block update in Eq.~\eqref{eq:distributed_primal} can be written,
or upper-bounded for analysis, as
\begin{equation}
    y[k+1]
    =
    \proj_{\Y}
    \left[
        y[k]-\eta(\xi_k+e_k)
    \right],
    \label{eq:abstract_projected_update}
\end{equation}
where $\|\xi_k+e_k\|\le G_F$. Moreover, the aggregate gradient, field-tracking,
active-direction, and command-recovery errors satisfy
\begin{equation}
    E_T:=\sum_{k=1}^T \|e_k\|<\infty,
    \label{eq:error_budget}
\end{equation}
where $e_k$ is the discrepancy between the ideal centralized subgradient
$\xi_k$ and the aggregate direction actually produced by the defenders.
\end{assumption}
Assumption~\ref{ass:regret} connects the concrete recursion to the regret
discussion. In particular, if the direction-agreement routine returns a common
$\nu_k\in\mathcal D_k(y[k])$ and the aggregate field is known exactly,
Eq.~\eqref{eq:concrete_direction} gives the corresponding block of the
centralized subgradient. The distributed implementation perturbs this
ideal direction through three mechanisms:
\begin{equation}
    e_k
    =
    e_k^{\mathrm{net}}
    +
    e_k^{\mathrm{dir}}
    +
    e_k^{\mathrm{loc}}.
    \label{eq:error_decomposition}
\end{equation}
Here $e_k^{\mathrm{net}}$ is caused by imperfect tracking of the aggregate
field, $e_k^{\mathrm{dir}}$ by approximate agreement on the active
direction, and $e_k^{\mathrm{loc}}$ by local modeling, linearization, and
command-recovery errors.
If the local gradients are $L_s$-Lipschitz with respect to the
aggregate-field estimate, then
\begin{equation}
    \|e_k^{\mathrm{net}}\|
    \le
    L_s D_s[k].
    \label{eq:network_error_bound}
\end{equation}
For the score-consensus direction routine, let us define
\begin{equation}
    E_T^\nu
    :=
    L_{\mathrm{dir}}
    \sum_{k=1}^T
    \left(
        L_v\delta+2\epsilon_{\nu,k}
    \right),
    \label{eq:direction_error_budget}
\end{equation}
where $L_{\mathrm{dir}}$ is a sensitivity constant mapping the
directional-margin error in Eq.~\eqref{eq:direction_approximation_error} into
subgradient error. Under exact active-direction selection, we have $E_T^\nu=0$.
Combining Eq.~\eqref{eq:network_error_bound} with Theorem~\ref{prop:field_tracking} gives the explicit sufficient condition
\begin{equation}
    E_T
    \le
    E_T^{\mathrm{loc}}
    +
    E_T^\nu
    +
    \frac{L_s}{1-\theta}
    \left(
        D_s[0]+V_T^r
    \right),
    \label{eq:explicit_error_budget}
\end{equation}
where $E_T^{\mathrm{loc}}:=\sum_{k=1}^T\|e_k^{\mathrm{loc}}\|$. Hence, the
core message is that oracle tracking remains sublinear when all four
contributors are sublinear, i.e., oracle path variation, local implementation
error, direction-agreement error, and field-variation budget.
The scalar $\eta$ in Eq.~\eqref{eq:abstract_projected_update} is the primal
step size of the aggregate projected-gradient model. In the concrete
distributed recursion, it corresponds to the common primal gain
$\eta_y$ in \eqref{eq:distributed_primal}, after collecting the block
updates into the stacked vector $y[k]$. Notice that larger values make the defenders
react more aggressively to the current active-direction loss, while
smaller values make the motion more conservative. The regret bound below
is valid for every constant $\eta>0$; the displayed choice balances the
tracking term, which scales as $1/\eta$, with the accumulated gradient
term which scales as $\eta T$.

Let us first establish a one-step inequality for the projected update, which is
an ancillary technical result that will be used in the regret proof.
\begin{lemma}
\label{lem:projection}
Under Assumption~\ref{ass:regret}, for any $z\in\Y$ and every $k$, it holds
\begin{align}
    \langle \xi_k,y[k]-z\rangle
    &\le
    \frac{\|y[k]-z\|^2}{2\eta}
    -
    \frac{\|y[k+1]-z\|^2}{2\eta}
    \notag\\
    &\quad
    +
    \frac{\eta G_F^2}{2}
    +
    D\|e_k\|.
    \label{eq:projection_lemma}
\end{align}
\end{lemma}

\begin{proof}
Let us define $d_k:=\xi_k+e_k$. By
Eq.~\eqref{eq:abstract_projected_update}, we have
$y[k+1]=\proj_{\Y}(y[k]-\eta d_k)$. Using the nonexpansiveness of Euclidean
projection onto the closed convex set $\Y$, it follows that
\begin{align}
    \|y[k+1]-z\|^2
    &\le
    \|y[k]-\eta d_k-z\|^2
    \notag\\
    &=
    \|y[k]-z\|^2
    \notag\\
    &\quad
    -
    2\eta\langle d_k,y[k]-z\rangle
    +
    \eta^2\|d_k\|^2 ,
\end{align}
and rearranging terms gives
\begin{align}
    \langle d_k,y[k]-z\rangle
    &\le
    \frac{\|y[k]-z\|^2}{2\eta}
    -
    \frac{\|y[k+1]-z\|^2}{2\eta}
    \notag\\
    &\quad
    +
    \frac{\eta\|d_k\|^2}{2}.
\end{align}
Since $\xi_k=d_k-e_k$, we can write
\begin{equation}
    \langle \xi_k,y[k]-z\rangle
    =
    \langle d_k,y[k]-z\rangle
    -
    \langle e_k,y[k]-z\rangle .
\end{equation}
Finally, using $\|d_k\|\le G_F$ from
Assumption~\ref{ass:regret} and $\|y[k]-z\|\le D$ from the diameter of
$\Y$, we obtain Eq.~\eqref{eq:projection_lemma}. This completes our proof.
\end{proof}
We are now in position to prove the main dynamic-regret bound.
\begin{theorem}
\label{thm:dynamic_regret}
Under Assumptions~\ref{ass:graph} and~\ref{ass:regret}, a projected
distributed online update with constant step size $\eta>0$ satisfies
\begin{equation}
    \operatorname{Reg}_T^{\operatorname{dyn}}
    \le
    \frac{D^2+2DP_T}{2\eta}
    +
    \frac{\eta G_F^2T}{2}
    +
    D E_T .
    \label{eq:dynamic_regret_bound}
\end{equation}
If an a priori bound $\bar P_T\ge P_T$ is available, the choice
\begin{equation}
    \eta
    =
    \frac{\sqrt{D^2+2D\bar P_T}}{G_F\sqrt{T}}
    \label{eq:stepsize_choice}
\end{equation}
gives
\begin{equation}
    \operatorname{Reg}_T^{\operatorname{dyn}}
    =
    O\left(\sqrt{T(1+\bar P_T)}+E_T\right).
    \label{eq:dynamic_regret_rate}
\end{equation}
Hence the average regret vanishes whenever $\bar P_T=o(T)$ and
$E_T=o(T)$. If the path-variation budget is not known in advance, the
bound in Eq.~\eqref{eq:dynamic_regret_bound} still holds for any chosen
$\eta>0$; adaptive or doubling choices can be used to estimate the
appropriate scale online.
\end{theorem}

\begin{proof}
By convexity of $F_k$ and $\xi_k\in\partial F_k(y[k])$, we have the following 
subgradient inequality
\begin{equation}
    F_k(y[k])-F_k(y_k^\star)
    \le
    \langle \xi_k,y[k]-y_k^\star\rangle.
\end{equation}
Applying Lemma~\ref{lem:projection} with $z=y_k^\star$ gives
\begin{align}
    F_k(y[k])-F_k(y_k^\star)
    &\le
    \frac{a_k-b_k}{2\eta}
    +
    \frac{\eta G_F^2}{2}
    +
    D\|e_k\|,
    \label{eq:regret_one_step}
\end{align}
where we define
\[
    a_k:=\|y[k]-y_k^\star\|^2,\qquad
    b_k:=\|y[k+1]-y_k^\star\|^2 .
\]
In order to handle the moving comparator $y_k^\star$, we compare consecutive squared
distances. Because all iterates and comparators lie in $\Y$, whose diameter
is $D$, we obtain the one-step bound as
\begin{align}
    a_{k+1}-b_k
    &=
    \|y[k+1]-y_{k+1}^\star\|^2
    -
    \|y[k+1]-y_k^\star\|^2
    \notag\\
    &\le
    2D\|y_{k+1}^\star-y_k^\star\|.
    \label{eq:moving_comparator_step}
\end{align}
In particular, the last inequality follows from
$|\|p-q\|^2-\|p-r\|^2|
\le \|q-r\|(\|p-q\|+\|p-r\|)$ with
$p=y[k+1]$, $q=y_{k+1}^\star$, and $r=y_k^\star$.
Therefore,
\begin{equation}
    a_k-b_k
    \le
    a_k-a_{k+1}
    +
    2D\|y_{k+1}^\star-y_k^\star\|.
\end{equation}
Summing Eq.~\eqref{eq:regret_one_step} from $k=1$ to $T$ and using the
previous relation yields
\begin{align}
    \operatorname{Reg}_T^{\operatorname{dyn}}
    &\le
    \frac{a_1-a_{T+1}}{2\eta}
    +
    \frac{D}{\eta}
    \sum_{k=1}^{T-1}\|y_{k+1}^\star-y_k^\star\|
    \notag\\
    &\quad+
    \frac{\eta G_F^2T}{2}
    +
    D\sum_{k=1}^T\|e_k\|.
\end{align}
Since $0\le a_1\le D^2$ and $a_{T+1}\ge 0$, we obtain
Eq.~\eqref{eq:dynamic_regret_bound}. If $\bar P_T\ge P_T$ and
$\eta$ is chosen as in Eq.\eqref{eq:stepsize_choice}, then we have
\begin{equation}
    \frac{D^2+2DP_T}{2\eta}
    +
    \frac{\eta G_F^2T}{2}
    \le
    G_F\sqrt{T(D^2+2D\bar P_T)}.
\end{equation}
which gives Eq.~\eqref{eq:dynamic_regret_rate}. This completes our proof.
\end{proof}

The next corollary specializes the theorem by making the network dependence
explicit through the tracking and direction-agreement error budgets.
\begin{corollary}
\label{cor:network_regret}
Let us assume that the hypotheses of Theorem~\ref{thm:dynamic_regret} hold and the
network error satisfies
Eq.~\eqref{eq:network_error_bound}. Then
\begin{align}
    \operatorname{Reg}_T^{\operatorname{dyn}}
    &\le
    \frac{D^2+2DP_T}{2\eta}
    +
    \frac{\eta G_F^2T}{2}
    +
    D E_T^{\mathrm{loc}}
    +
    D E_T^\nu
    \notag\\
    &\quad+
    \frac{D L_s}{1-\theta}
    \left(
        D_s[0]+V_T^r
    \right).
    \label{eq:network_explicit_regret}
\end{align}
Consequently, with $\eta$ chosen as in \eqref{eq:stepsize_choice} for an
available bound $\bar P_T\ge P_T$, the average dynamic regret vanishes if
\begin{align}
    \bar P_T=o(T),\qquad
    E_T^{\mathrm{loc}}=o(T),
    \qquad
    E_T^\nu=o(T),
    \notag\\
    V_T^r=o(T).
\end{align}
\end{corollary}

\begin{proof}
Starting from the decomposition in Eq.~\eqref{eq:error_decomposition}, we have
\[
    E_T
    \le
    E_T^{\mathrm{loc}}
    +
    E_T^\nu
    +
    \sum_{k=1}^T \|e_k^{\mathrm{net}}\|.
\]
Using Eq.~\eqref{eq:network_error_bound} and the cumulative tracking bound
in Eq.~\eqref{eq:cumulative_tracking_bound} yields
Eq.~\eqref{eq:explicit_error_budget}. Substituting this estimate for $E_T$
into Eq.~\eqref{eq:dynamic_regret_bound} gives
Eq.~\eqref{eq:network_explicit_regret}. The stated sublinear condition then
follows from Eq.~\eqref{eq:dynamic_regret_rate}.
\end{proof}
Finally, the following corollary translates the regret estimate into a direct bound on
the cumulative containment violation.
\begin{corollary}
\label{cor:violation_from_regret}
Let us assume that $F_k(y)\ge \alpha\Psi_k(y)$ for all $y\in\Y_k$ and all $k$, and let us 
define $B_T^\star:=\sum_{k=1}^T F_k(y_k^\star)$. Then
\begin{equation}
    \frac{V_T^\Psi}{T}
    \le
    \frac{1}{\alpha T}
    \left(
        \operatorname{Reg}_T^{\operatorname{dyn}}
        +
        B_T^\star
    \right).
    \label{eq:violation_from_regret}
\end{equation}
In particular, if $B_T^\star=o(T)$ and
$\operatorname{Reg}_T^{\operatorname{dyn}}=o(T)$, then the average
worst-case containment violation vanishes. The proof is complete.
\end{corollary}

\begin{proof}
From $F_k(y)\ge\alpha\Psi_k(y)$, summing over $k=1,\ldots,T$ gives
\begin{equation}
    \alpha V_T^\Psi
    \le
    \sum_{k=1}^T F_k(y[k]).
\end{equation}
By the definition of dynamic regret, it holds
\begin{equation}
    \sum_{k=1}^T F_k(y[k])
    =
    \operatorname{Reg}_T^{\operatorname{dyn}}
    +
    \sum_{k=1}^T F_k(y_k^\star).
\end{equation}
Combining the two relations and dividing by $\alpha T$ yields
Eq.~\eqref{eq:violation_from_regret}. The proof is complete.
\end{proof}
\color{black}

\providecommand{\simulationDataPath}{data}

\section{Simulation Study}\label{sec:simulation}

We consider a two-dimensional instance with $n=5$ defenders and one hijacked target. The admissible set $\Omega_k$ is a moving disk of radius 1.15, the adversarial command set is a disk of radius 0.07, and the defenders have maximum displacement 0.20 per time step. The target starts at the center of $\Omega_0$, while the defenders are initialized on a loose support ring around the disk. All simulations use the random seed 11.

The centralized controller solves the one-step oracle problem by projected finite-difference descent. The distributed controllers use dynamic average consensus for the aggregate field. The first distributed controller selects the active direction by score consensus over a finite direction dictionary; the second uses the cheaper projection direction. The regularization includes motion, clearance, and a light coverage term that keeps the defenders distributed around the moving disk unless the containment loss calls for a stronger local response.

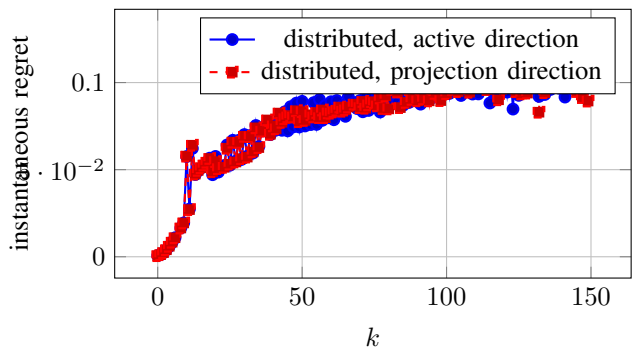
\begin{figure}[t]\centering
\begin{tikzpicture}
\begin{axis}[width=0.95\columnwidth,height=0.58\columnwidth,grid=both,xlabel={$k$},ylabel={instantaneous regret},legend pos=north east]
\addplot+[thick] table[x=k,y=active_instant_regret,col sep=comma]{\simulationDataPath/regret_timeseries_downsampled.csv};
\addlegendentry{distributed, active direction}
\addplot+[thick,dashed] table[x=k,y=projection_instant_regret,col sep=comma]{\simulationDataPath/regret_timeseries_downsampled.csv};
\addlegendentry{distributed, projection direction}
\end{axis}
\end{tikzpicture}
\caption{Instantaneous dynamic regret against the centralized oracle.}
\label{fig:sim_instant_regret}
\end{figure}

\begin{figure}[t]\centering
\begin{tikzpicture}
\begin{axis}[width=0.95\columnwidth,height=0.58\columnwidth,grid=both,xlabel={$k$},ylabel={cumulative regret},legend pos=north west]
\addplot+[thick] table[x=k,y=active_cumulative_regret,col sep=comma]{\simulationDataPath/regret_timeseries_downsampled.csv};
\addlegendentry{distributed, active direction}
\addplot+[thick,dashed] table[x=k,y=projection_cumulative_regret,col sep=comma]{\simulationDataPath/regret_timeseries_downsampled.csv};
\addlegendentry{distributed, projection direction}
\end{axis}
\end{tikzpicture}
\caption{Cumulative dynamic regret against the centralized oracle.}
\label{fig:sim_regret}
\end{figure}
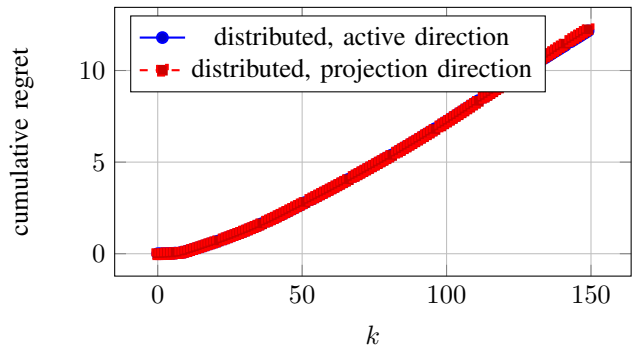

\begin{figure}[t]\centering
\begin{tikzpicture}
\begin{axis}[width=0.95\columnwidth,height=0.58\columnwidth,grid=both,xlabel={$k$},ylabel={average regret},legend pos=north east]
\addplot+[thick] table[x=k,y expr=\thisrow{active_cumulative_regret}/(\thisrow{k}+1),col sep=comma]{\simulationDataPath/regret_timeseries_downsampled.csv};
\addlegendentry{distributed, active direction}
\addplot+[thick,dashed] table[x=k,y expr=\thisrow{projection_cumulative_regret}/(\thisrow{k}+1),col sep=comma]{\simulationDataPath/regret_timeseries_downsampled.csv};
\addlegendentry{distributed, projection direction}
\end{axis}
\end{tikzpicture}
\caption{Average dynamic regret. A decreasing trend indicates that cumulative regret grows sublinearly over the simulated horizon.}
\label{fig:sim_average_regret}
\end{figure}
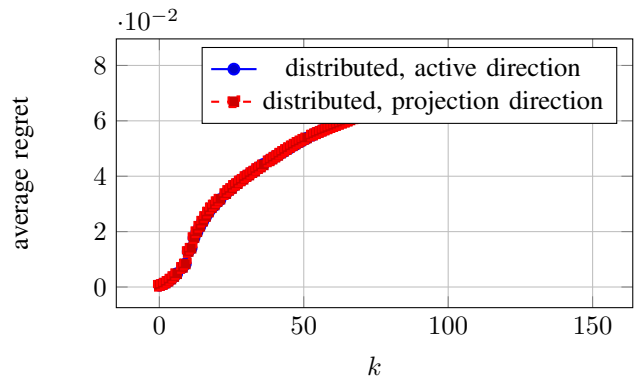

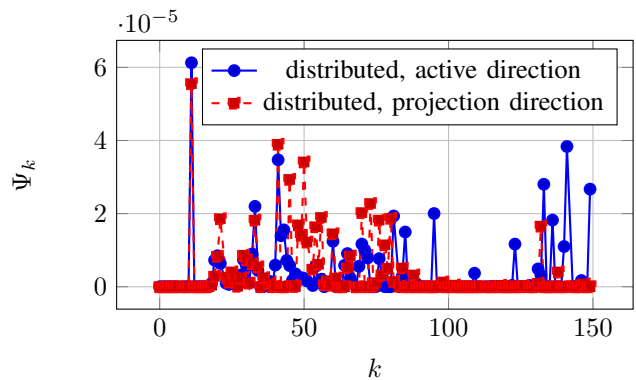
\begin{figure}[t]\centering
\begin{tikzpicture}
\begin{axis}[width=0.95\columnwidth,height=0.58\columnwidth,grid=both,xlabel={$k$},ylabel={$\Psi_k$},legend pos=north east]
\addplot+[thick] table[x=k,y=active_violation,col sep=comma]{\simulationDataPath/violation_timeseries_downsampled.csv};
\addlegendentry{distributed, active direction}
\addplot+[thick,dashed] table[x=k,y=projection_violation,col sep=comma]{\simulationDataPath/violation_timeseries_downsampled.csv};
\addlegendentry{distributed, projection direction}
\end{axis}
\end{tikzpicture}
\caption{Worst-case one-step containment violation.}
\label{fig:sim_violation}
\end{figure}

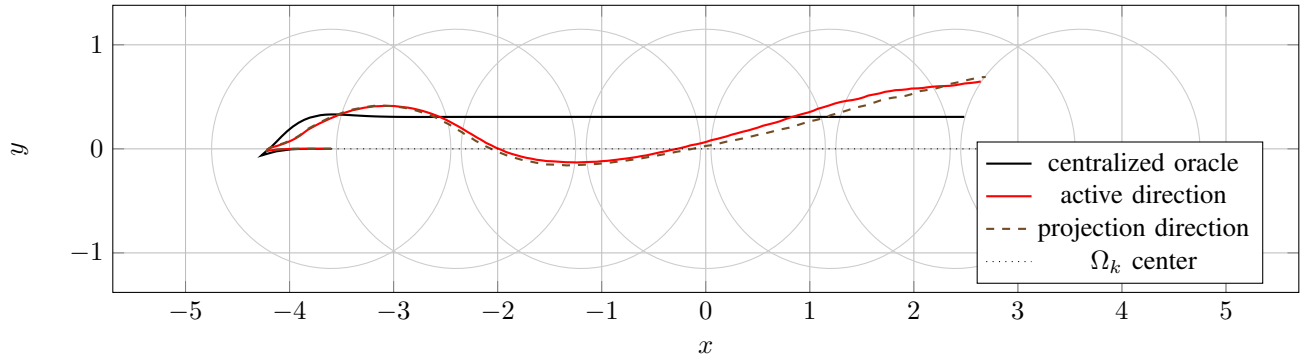
\begin{figure*}[t]\centering
\begin{tikzpicture}
\begin{axis}[width=1.95\columnwidth,height=1.95\columnwidth,grid=both,xlabel={$x$},ylabel={$y$},axis equal image,legend pos=south east]
\addplot[black!20,no markers,domain=0:360,samples=80,forget plot] ({-3.600+1.150*cos(x)},{0.000+1.150*sin(x)});
\addplot[black!20,no markers,domain=0:360,samples=80,forget plot] ({-2.400+1.150*cos(x)},{0.000+1.150*sin(x)});
\addplot[black!20,no markers,domain=0:360,samples=80,forget plot] ({-1.200+1.150*cos(x)},{0.000+1.150*sin(x)});
\addplot[black!20,no markers,domain=0:360,samples=80,forget plot] ({0.000+1.150*cos(x)},{0.000+1.150*sin(x)});
\addplot[black!20,no markers,domain=0:360,samples=80,forget plot] ({1.200+1.150*cos(x)},{0.000+1.150*sin(x)});
\addplot[black!20,no markers,domain=0:360,samples=80,forget plot] ({2.400+1.150*cos(x)},{0.000+1.150*sin(x)});
\addplot[black!20,no markers,domain=0:360,samples=80,forget plot] ({3.600+1.150*cos(x)},{0.000+1.150*sin(x)});
\addplot+[thick,black,no markers] table[x=target_x,y=target_y,col sep=comma]{\simulationDataPath/trajectory_centralized_downsampled.csv};
\addlegendentry{centralized oracle}
\addplot+[thick,no markers] table[x=target_x,y=target_y,col sep=comma]{\simulationDataPath/trajectory_distributedActive_downsampled.csv};
\addlegendentry{active direction}
\addplot+[thick,dashed,no markers] table[x=target_x,y=target_y,col sep=comma]{\simulationDataPath/trajectory_distributedProjection_downsampled.csv};
\addlegendentry{projection direction}
\addplot+[black,dotted,no markers] table[x=omega_x,y=omega_y,col sep=comma]{\simulationDataPath/omega_downsampled.csv};
\addlegendentry{$\Omega_k$ center}
\end{axis}
\end{tikzpicture}
\caption{Target trajectories and moving admissible-set centers. The MATLAB video overlays the full admissible disks and defender positions.}
\label{fig:sim_trajectories}
\end{figure*}



\section{Conclusion}
\label{sec:conclusion}

This paper develops a cleaner formulation of indirect containment for a
compromised target that retains the same repulsive low-level safety
behavior as the defenders. The defenders choose commands, these commands
induce configurations, and the configurations induce an aggregate field on
the target. The resulting interaction is naturally described as an online
Stackelberg game with distributed information. The central research
question becomes whether the defenders can track the centralized
Stackelberg oracle with sublinear regret and maintain small
worst-case containment violation.







\end{document}